\documentclass[11pt]{article}

\usepackage[authordate,backend=biber,uniquename=false]{biblatex-chicago}
\usepackage{algorithm, algpseudocode, amsfonts, amssymb, amsmath, amsthm, booktabs, graphicx, bm, float, enumitem, csquotes, pifont, fancyhdr, setspace, siunitx, lipsum, multirow, comment, hyperref, caption, subcaption, tabularx, array, makecell}
\usepackage[margin=1.25in]{geometry}
\usepackage[dvipsnames]{xcolor}
\usepackage[colorinlistoftodos]{todonotes}
\usepackage{tikz}
\usetikzlibrary{arrows.meta,bending,
                decorations.markings,graphs}

\newtheorem{proposition}{Proposition}
\newtheorem{theorem}{Theorem}

\newtheorem{definition}{Definition}
\newtheorem{fact}{Fact}

\newtheorem{condition}{Condition}
\newtheorem{lemma}{Lemma}

\DeclareFieldFormat{url}{}
\DeclareFieldFormat{urldate}{}
\DeclareFieldFormat{note}{}
\DeclareFieldFormat{issn}{}
\DeclareFieldFormat{doi}{}

\hypersetup{
    colorlinks=true,
    linkcolor=blue,
    filecolor=magenta,      
    urlcolor=blue,
    citecolor=blue
    }

\title{The Empirical Welfare Content of International Price and Income Comparisons}
\author{Hubert Wu\thanks{Department of Economics and Nuffield College, University of Oxford; \href{mailto:hubert.wu@nuffield.ox.ac.uk}{hubert.wu@nuffield.ox.ac.uk}. I thank Ian Crawford, Erwin Diewert, John Quiggan, Diane Coyle, Robert C. Feenstra, Giulio Gottardo, Alex Teytelboym, Bert M. Balk, participants at the 2025 RES Annual Conference, 2025 TADC, 2024 UNSW-ESCoE Conference on Economic Measurement, and seminar participants at the University of Oxford for comments, suggestions, and other helpful input. I also thank members of the World Bank's International Comparison Program, Development Data Group team for providing access to the data used in this paper's empirical applications.}}
\date{9 December 2025}

\begin{document}
{
\makeatletter
\addtocounter{footnote}{1} 
\renewcommand\thefootnote{\@fnsymbol\c@footnote}%
\makeatother
\maketitle
}
\begin{abstract}
\noindent Multilateral index numbers are often used to make claims about welfare, such as treating PPPs as cross-country costs of living or real incomes as indicators of living standards. However, such interpretations may not be consistent with the observed data. To study this problem, I derive multilateral bounds on welfare implied by revealed preference and use these to appraise leading comparison methods. My findings support the welfare-interpretability of the contemporary indices I examine, but not of market exchange rates. When using a welfare-consistent multilateral index, the world in 2017 appears larger and more equal vis-à-vis the United States than conventional measures.
\end{abstract}


\newpage


\section{Introduction} \label{sec:intro}

Cross-country comparisons of prices and incomes are fundamental within economics. Such assessments underpin our understanding of issues such as the relative size of economies, the fraction of the world’s population living in extreme poverty, wealth inequality among nations, and variations in economic growth rates and convergence in living standards. As is well-known however, it is not appropriate to conduct such calculations with either market exchange rates (which do not hold purchasing power constant across currencies) or classical bilateral index number formulae (which in settings involving more than two countries produce produce a multiplicity of index values for any given country pair).

A substantial literature developing and studying multilateral indices has emerged to solve this problem. For example, coherent comparisons can be constructed by considering a restricted set of comparisons through a single country (e.g. \textcite{geary_note_1958}, \textcite{khamis_new_1972}), by averaging permutations of country networks (e.g. \textcite{gini_quelques_1924}; \textcite{elteto_problem_1964}; and \textcite{szulc_indices_1964}), or from constructing minimum spanning trees (\cite{hill1999comparing}; \cite{hill2004constructing}). Consumer theory can also be brought to bear on the issue. In the ‘economic approach’ of index number theory (\cite{diewert_axiomatic_1999}), multilateral indices arise from a cost-of-living index with the assumptions of homotheticity or a fixed indifference base. Index formulae of this kind can be derived from parametric models (e.g. \textcite{cavallo2023product}) or possess the property of being `superlative' (\cite{diewert_exact_1976}) in the sense that under the appropriate assumptions they approximate a broad class of preferences to the second order.

Despite the existence of many comparison methods, a striking and uncomfortable fact lies behind them all: welfare claims based on these comparisons may be rejected by the data. Rather than being choice-consistent objects derived from behaviour, multilateral indices are statistical formulae designed to solve an aggregation and comparison problem. Even when indices possess a link to consumer theory, strong assumptions about the form of country tastes such as a single global preference relation, homotheticity, or knowledge of the explicit form of the utility function are required for a cost-of-living interpretation, with behavioural assumptions used as a means to construct or justify an index rather than being objects subject to critique. Such assumptions have long been known to be empirically suspect\footnote{For example, \textcite{deaton_understanding_2010} reflect that ``The assumption that all countries have identical homothetic tastes [is] a proposition that has been falsified by more than 150 years
of empirical demand analysis... Critics of
[a non-economic] approach [to constructing multilateral indices] argue that it leaves the welfare basis of the calculations unclear. We agree, but do not know how to do better (p. 6)''.}. It is thus an open question as to what extent, if at all, multilateral indices generate comparisons that are empirically consistent with economic welfare, and also how indices perform comparatively in this regard. This is a concerning situation because many applications of international price and income comparisons, such as measuring global living standards or disbursing foreign aid, inherently depend on a well-defined notion of economic welfare to be meaningful exercises.

This issue motivates the central question of this paper: what can one infer about the welfare content of multilateral comparisons of prices and incomes, given the diverse set of motivations, methods, assumptions, and tastes that underlie them? I answer this question by deriving the conditions that any index must satisfy for a given dataset to admit a cost of living. I find these conditions by exploiting the concept of a reference consumer (\textcite{neary_rationalizing_2004} and \textcite{crawford2008testing}), which I define and characterise as a set of tastes that rationalise the pooled data of multiple countries. The reference consumer can be thought of as a hypothetical individual who when faced with the relative prices in each of the world's countries would purchase their observed consumption bundles. These tastes in turn imply a set of sharp non-parametric bounds on the multilateral cost of living implied by revealed preference. The bounds improve upon their classical bilateral analogues and permit an appraisal of the welfare-validity of comparison methods that does not necessitate a separate index or form of preferences as an external benchmark. The test can also examine  assumptions about the form of country preferences such as homotheticity and identical cross-country tastes which are foundational to nearly all multilateral indices related to the cost of living. Finally, the bounds define a novel multilateral price index that generalises the star system of international comparisons while accounting for income effects.

My main contributions are empirical. In my first major application, I appraise the interpretability of several multilateral indices which are widely used in international comparisons in a scheme analogous to treatment condition selection in an experiment. I find that these indices are largely consistent with a cost-of-living interpretation based on the data. All the examined indices outperform market exchange rates, providing a further endorsement of index number methods that imbibe granular data about goods and services. Moreover, I find that superlative indices (\cite{diewert_exact_1976}) — ones that approximate an underlying utility function under appropriate assumptions — are frequently more empirically consistent with a multilateral cost-of-living index than non-superlative methods. However, the magnitude by which all the appraised indices exceed the permitted bounds is modest. My design also permits a novel test of the path-averaging `GEKS step' that transforms bilateral index numbers into transitive multilateral ones. I find that the effect of this procedure is small in terms of welfare validity. In the course of these findings I also document the surprising fact that a reference consumer with non-homothetic tastes accounts for nearly all output in 2017. This fact provides justification for approaches such as \textcite{diewert_exact_1976} and \textcite{diewert_axiomatic_1999} which assume global optimising behaviour in multilateral  settings based on a single preference relation and which feature widely in contemporary applications. 

Moving beyond index properties to the uses of international price and income comparisons themselves, in the paper's second application I re-calculate global output and cross-country income inequality using index values from the generalised star system based on the reference consumer and their cost-of-living bounds. I find that the world in 2017 was substantially larger and more equal than suggested by conventional measures at the time from the International Comparison Program. The practical magnitude of these differences is substantial and is roughly comparable to that arising from using PPPs instead of market exchange rates (for output), and between the pre-and post-tax Gini coefficient among OECD countries (for inequality). Despite these results, in recognition of the fact that different multilateral indices may suit different applications, in the Supplemental Appendix I also demonstrate how any desired index can be rendered welfare-consistent.

My paper is related to the literatures about index numbers, cross-country welfare, and empirical revealed preference. My main results  clarify the empirical welfare content of methods used to calculate PPPs and real incomes including those that are widely used today. This topic is closely related to the influential `economic approach' to index number theory pioneered by \textcite{diewert_exact_1976} and \textcite{diewert_axiomatic_1999} which derives and studies indices that are solutions to a utility maximisation or cost minimisation problem, and which therefore also have a cost of living interpretation. My paper is part of this literature, but alters the roles of data and theory relative to much of it as a result of the revealed preference approach I take and the philosophy that underpins it. My first application critiques existing multilateral indices which is a topic where an extensive literature exists: see for example \textcite{ackland2013measuring}, \textcite{diewert_methods_2013}, \textcite{deaton_understanding_2010}, \textcite{neary_rationalizing_2004}, \textcite{alterman_international_1999}, \textcite{diewert_axiomatic_1999}, and \textcite{hill_measuring_2000} among others. My paper complements works like these by conducting an appraisal based on economic theory that also does not require assuming a given index or utility function is true. (A notable exception is \textcite{ackland2013measuring}, which contains a comparison of two indices against the homothetic counterpart to my bounds.) Regarding cross-country welfare, my re-calculation of world output and inequality is related to a literature about new ways of measuring welfare between countries. Recent papers in this realm include \textcite{cavallo2023product}, who exploit the structure of the \textcite{melitz2003impact} model to measure the cross-country cost of living); \textcite{argente2023measuring}, who use new microdata to measure the cost of living between Mexico and the United States); and \textcite{jones2016beyond}, who derive a welfare statistic incorporating factors such as leisure, mortality, and inequality. Turning to my methods, three  papers that pioneered the use of revealed preference theory in international comparisons are \textcite{dowrick_international_1994},  \textcite{dowrick_true_1997}, and \textcite{crawford2008testing}. Of these \textcite{dowrick_true_1997} is most closely related to mine as they construct a multilateral quantity index called the Ideal Afriat Index based on the homothetic version of Afriat's Theorem (\cite{diewert1973afriat}; \cite{afriat_constructability_1981}; \cite{varian_non-parametric_1983}) and use it to study substitution bias and income convergence in 1980 for 17 countries. Despite using similar theoretical tools, I focus on a different research question, welfare object, and set of applications. Finally, \textcite{crawford2008testing} define a non-parametric reference consumer which is a critical enabling concept for my results but is different to my paper's core contributions.

The remainder of this paper is structured as follows. Section \ref{sec:intlrefcons} explains the problem of international comparisons' welfare-interpretability, derives the revealed preference bounds on the cross-country cost of living, and then defines a new index based on these ideas. Section \ref{sec:application} presents my main results. Section \ref{sec:discussion} discusses and concludes.

\section{Welfare-interpretability and the reference consumer} \label{sec:intlrefcons}

In this section I illustrate the problem of international comparisons' welfare-interpretability and explain how the perspective of a reference consumer can help study this problem. The section's main conceptual tools are multilateral revealed preference bounds on the cost of living (\ref{subsec:mult_bounds}) and an index based on these bounds that generalises the star system (\ref{subsec:fat_star_og}).

\subsection{Setting} \label{sec:notation}

The paper's setting involves $N$ countries with consumption over $K$ goods. Let $\mathbf{p}_{n}\equiv (p_{n}^{1},...,p_{n}^{K})$ and $\mathbf{q}_{n}\equiv (q_{n}^{1},...,q_{n}^{K})$ denote the prices and per-capita quantities for the $k=1,...,K$ goods for country $n$. The consumption space is $\mathcal{Q}\subseteq \mathbb{R}_{+}^{K}$. Prices $\mathbf{p}_{n}\in \mathbb{R}_{++}^{K}$ are in the units of country $n$'s local currency.

The fundamental welfare object I study is the \textcite{konus_problem_1939} cost of living, which is defined through the money metric  $e(\mathbf{p},u(\mathbf{q}))$ for a fixed indifference level $u(\mathbf{q})$. As well as being a cardinalisation of utility, the money metric is `canonical' (\cite{knoblauch_tight_1992}) in the sense that any utility function with the same family of indifference curves produces the same money metric. Expenditure functions are the building blocks of the cost of living for countries $i$ and $j$,
\begin{align} 
    K_{ij} (\mathbf{p}_{i},\mathbf{p}_{j},u(\mathbf{q})) \equiv \frac{e(\mathbf{p}_{j},u(\mathbf{q}))}{e(\mathbf{p}_{i},u(\mathbf{q}))}.
    \label{eq:konus_col}
\end{align}
\noindent As well as having a welfare interpretation grounded in consumer theory, country indices based on \eqref{eq:konus_col} are also resistant to substitution bias that arises in the form of the Gerschenkron effect of \textcite{gerschenkron1947soviet}. Cost-of-living indices can be rendered transitive for multilateral purposes by fixing an indifference base $u(\mathbf{q})$ or by assuming homothetic tastes which allows \eqref{eq:konus_col} to be written in the form $f(\mathbf{p}_{j})/f(\mathbf{p}_{i})$ which is independent of $u(\mathbf{q})$.

\subsection{Multilateral indices' welfare interpretability: an example} \label{subsec:motivating_example}

To ground the paper's motivation, Figure \ref{fig:minimal_example_1} supplies a simple illustration of the welfare-interpretability problem underlying multilateral indices using standard graphs in the field (e.g. see \textcite{hill2004constructing}). Consider computing the cost of living among four countries $A$, $B$, $C$, and $D$ for the example data. As seen in in panel (a)'s graph, for countries $A$ and $B$, the bilateral Fisher index is non-transitive (as the direct path (in cyan) $F_{AB}=1.004$ differs from the indirect $F_{AC}F_{CB}=1.085$) which renders multilateral comparisons incoherent. Taking the geometric mean of possible paths produces the `GEKS' parities of \textcite{gini_quelques_1924}; \textcite{elteto_problem_1964}; and \textcite{szulc_indices_1964} in panel (b). The GEKS procedure is widely used in contemporary applications and generates the unique, path-independent value for $G_{AB}=G_{AC}G_{CB}=1.030$. However, for the example, the GEKS parities are uninterpretable from the perspective of economic welfare (i.e. as a cost of living) because the valuation for countries $B$ and $C$ $G_{BC}=0.719$ (the red edge in (b)) exceeds the Laspeyres upper bound on the cost-of-living $L_{BC}=0.7$ implying the contradiction $e(\mathbf{p}_{C},u(\mathbf{q}_{B}))>\mathbf{p}_{C}\mathbf{q}_{B}$. One solution to this issue is to use a `star' system of bilateral Laspeyres indices in panel (c) with C featuring as the star's `hub'. This system achieves a transitive and welfare-interpretable set of parities with $L_{AB}=L_{CA}^{-1}L_{CB}=0.960$, but the arbitrarily-esteemed position of $C$'s preferences as a base for all comparisons lacks a compelling rationale in multilateral settings.

\begin{figure}[h!]
    \centering
    \captionsetup[subfigure]{justification=centering,font=small}

    \begin{subfigure}[t]{0.32\textwidth}
        \centering
        \resizebox{\linewidth}{!}{%
        \begin{tikzpicture}[
            vertex/.style={circle, draw, fill=black, inner sep=0pt, minimum size=6pt},
            blue vertex/.style={vertex, draw=blue, fill=blue},
            cyan vertex/.style={vertex, draw=cyan, fill=cyan},
            edge/.style={-{Stealth[length=2.4mm]}, bend left=20},
            blue edge/.style={edge, draw=blue},
            cyan edge/.style={edge, draw=cyan},
            weight/.style={font=\scriptsize, midway, fill=white, inner sep=1pt}
        ]
            \node[vertex] (a) at (0,0)       [label=left:$A$] {};
            \node[vertex] (b) at (3,0)       [label=right:$B$] {};
            \node[vertex] (c) at (1.5,2.6)   [label=above:$C$] {};
            \draw[cyan edge] (a) to node[weight] {$1.004$} (b);
            \draw[blue edge] (a) to node[weight] {$0.760$} (c);
            \draw[blue edge] (c) to node[weight] {$1.429$} (b);
        \end{tikzpicture}}%
        \caption{Fisher}
    \end{subfigure}
    \hfill
    \begin{subfigure}[t]{0.32\textwidth}
        \centering
        \resizebox{\linewidth}{!}{%
        \begin{tikzpicture}[
            vertex/.style={circle, draw, fill=black, inner sep=0pt, minimum size=6pt},
            edge/.style={-{Stealth[length=2.4mm]}, bend left=20, draw=ForestGreen},
            green edge/.style={edge, draw=ForestGreen},
            red edge/.style={edge, draw=red},
            weight/.style={font=\scriptsize, midway, fill=white, inner sep=1pt}
        ]
            \node[vertex] (a) at (0,0)       [label=left:$A$] {};
            \node[vertex] (b) at (3,0)       [label=right:$B$] {};
            \node[vertex] (c) at (1.5,2.6)   [label=above:$C$] {};
            \draw[green edge] (a) to node[weight] {$1.030$} (b);
            \draw[green edge] (a) to node[weight] {$0.740$} (c);
            \draw[green edge, bend right = 20] (c) to node[weight] {$1.392$} (b);
            \draw[red edge, bend right = 20] (b) to node[weight] {$0.719$} (c);
        \end{tikzpicture}}%
        \caption{GEKS}
    \end{subfigure}
    \hfill
    \begin{subfigure}[t]{0.32\textwidth}
        \centering
        \resizebox{\linewidth}{!}{%
        \begin{tikzpicture}[
            vertex/.style={circle, draw, fill=black, inner sep=0pt, minimum size=6pt},
            edge/.style={-{Stealth[length=2.4mm]}, bend left=20, draw=Purple},
            purple edge/.style={edge, draw=Purple},
            red edge/.style={edge, draw=red},
            weight/.style={font=\scriptsize, midway, fill=white, inner sep=1pt}
        ]
            \node[vertex] (a) at (0,0)       [label=left:$A$] {};
            \node[vertex] (b) at (3,0)       [label=right:$B$] {};
            \node[vertex] (c) at (1.5,2.6)   [label=above:$C$] {};
            \draw[purple edge] (c) to node[weight] {$1.489$} (a);
            \draw[purple edge, bend right=20] (c) to node[weight] {$1.429$} (b);
        \end{tikzpicture}}%
        \caption{Single-base Laspeyres}
    \end{subfigure}

    \caption{\textbf{Conceptual difficulties regarding multilateral comparisons of cross-country PPPs}. The data for this example are $ \mathbf{p}_{A}\!\equiv\!(5,9),\; \mathbf{p}_{B}\!\equiv\!(7,7),\; \mathbf{p}_{C}\!\equiv\!(10,10),\; \mathbf{q}_{A}\!\equiv\!(8,6),\; \mathbf{q}_{B}\!\equiv\!(7,10),\; \mathbf{q}_{C}\!\equiv\!(1,9) $. Index values for each panel are the respective directed graph's edge weights. (a) The Fisher parities are path-dependent. (b) The GEKS parities are transitive but cannot be interpreted as a cost-of-living index compatible with the data. (c) The Laspeyres parities with country C as a base achieves transitivity and welfare-interpretability at the cost of using country C's tastes in all pairwise comparisons.}
    \label{fig:minimal_example_1}
\end{figure} 

\subsection{The international reference consumer} \label{sec:irc}

The third panel in Figure \ref{fig:minimal_example_1} highlights a major reason for why studying the interpretability issue with the cost of living has historically been difficult (for example, see the discussion in \textcite{Pollak1971a}, \textcite{collier_comment_1999}, and \textcite{neary_rationalizing_2004}): in a multilateral setting, there are many answers to the question `whose preferences should we use to conduct a comparison?' 

To begin surmounting this problem, I leverage the idea of a reference consumer (à la \textcite{neary_rationalizing_2004}) whose behaviour is consistent with the data of multiple countries and whose tastes can therefore be used as the indifference base for the money metric. Definition \ref{def:ref_consumer} states this idea in terms of a rationalising preference.

\begin{definition}[The reference consumer] \label{def:ref_consumer}
There exists a \emph{reference consumer} for the pooled country dataset
    \[
    \mathcal{D} := \left\{ \left(\mathbf{p}_{i}, \mathbf{q}_{i}\right) \mid i \in \{1, 2, \ldots, N\} \right\}.
    \]
    if there exists a utility function $u:\mathcal{Q}\rightarrow\mathbb{R}$
    where $u(\mathbf{q}_{i})\ge u(\mathbf{q})$ for all
    $\mathbf{q}\in \mathcal{Q}$ such that $\mathbf{p}_{i}\cdot\mathbf{q}\le\mathbf{p}_{i}\cdot\mathbf{q}_{i}$.
\end{definition}
The reference consumer is an individual whose behaviour is consistent with the collective data of multiple countries. Using these tastes to calculate the indifference base for a multilateral cost-of-living index is an attractive answer to the question of `whose preferences should we use?' because doing so obviates the problem of privileging a single country's preferences in all comparisons. In contrast to \textcite{neary_rationalizing_2004}'s concept, the tastes in Definition \ref{def:ref_consumer} do not require assumptions about the functional form of utility or econometric estimation from the data. Definition \ref{def:ref_consumer} also leads to improvements over classical  cost-of-living bounds, underpin an index appraisal scheme, and define a new non-homothetic index based on revealed preference. These benefits do not arise in \textcite{neary_rationalizing_2004}'s seminal characterisation.

How do we know whether a reference consumer exists for a given dataset? Answering this question is critical for investigating indices' welfare content so to do this, I now proceed to characterise the reference consumer's existence using graphs induced by calculable index numbers (as in the illustrative example at Figure \ref{fig:minimal_example_1}). I first review some basic definitions from graph theory required for the result. A \textit{weighted directed graph} is a triple $G=(V,E,w)$ comprised of a non-empty set $V$ of vertices, a set of directed edges $E\subseteq \{(\upsilon_{i},\upsilon_{j})|(\upsilon_{i},\upsilon_{j})\in V\times V \text{ and } \upsilon_{i}\neq \upsilon_{j}\}$, a weight function $w:E\rightarrow \mathbb{R}$ that assigns a real number to each  $(\upsilon_{i},\upsilon_{j})\in E$. A \textit{walk} in $G=(V,E,w)$ is a sequence of the form $\upsilon_{0} (\upsilon_{0},\upsilon_{1}) \upsilon_{1} (\upsilon_{1},\upsilon_{2}) \upsilon_{2}...(\upsilon_{k-1},\upsilon_{k}) \upsilon_{k}$ with $\upsilon_{i}\in V$ and $(\upsilon_{i},\upsilon_{i+1}) \in E$ for all $0\leq i<k$. A \textit{closed walk} in $G=(V,E,w)$ is a walk with $\upsilon_{0}=\upsilon_{k}$. Finally, a \textit{cycle} in $G=(V,E,w)$ is a closed walk with distinct vertices except for $\upsilon_{0}=\upsilon_{k}$. 

To proceed, consider the weighted directed graph $G=(V,E,w)$ associated with the pooled dataset $\mathcal{D}$ defined as follows. The set of vertices $V$ comprises observations with countries indexed by $\{1,...,N\}$. Each edge weight $w_{mn}$ is the bilateral Laspeyres quantity index for countries $m$ and $n$ as the reference and comparison periods,
    \begin{align}
        w_{mn}\equiv \frac{\mathbf{p}_{m}.\mathbf{q}_{n}}{\mathbf{p}_{m}.\mathbf{q}_{m}}
    \label{eq:paasche_edge_weight}
    \end{align}

With $G$ now defined we will consider the presence of certain cycles within it.
\begin{condition}[CEWEC]
    The \emph{Cycle Edge Weight Equality Condition (CEWEC)} is satisfied if, for every cycle in $G$ with vertices $\upsilon_1, \upsilon_2, \dots, \upsilon_k$ where the associated edge weights $w(\upsilon_\ell, \upsilon_{\ell+1}) \leq 1$ for $\ell = 1, 2, \dots, k$ (where $\upsilon_{k+1} = \upsilon_1$), it holds that $w(\upsilon_\ell, \upsilon_{\ell+1}) = 1$ for all $\ell$.
    \label{cond:cewec}
\end{condition}
The CEWEC arises in many graph-based settings within mechanism design, game theory, and finance. In the revealed preference literature it is equivalent to the `cyclical consistency' property of \textcite{afriat_construction_1967} typically applied to individual consumption datasets.

I am now ready to characterise the reference consumer and do so in Proposition \ref{prop:ref_consumer_characterisation}.

\begin{proposition} Consider the weighted directed graph $G=(V,E,w)$ associated with $\mathcal{D}$ with country vertices and Laspeyres (quantity) weights. The following statements are equivalent:
    \begin{enumerate}
        \item There exists a reference consumer for $\mathcal{D}$ whose tastes $u(\mathbf{q})$ are non-satiated, continuous, concave, and strictly monotonic.
        \item $G$ satisfies the CEWEC.
    \end{enumerate}
    \label{prop:ref_consumer_characterisation}
\end{proposition}
\noindent \textbf{Proof.} See \hyperref[proof:ref_consumer_characterisation]{\textbf{Appendix}}. \\

Proposition \ref{prop:ref_consumer_characterisation} is a graph-theoretic statement of Afriat's Theorem (\cite{afriat_construction_1967}) -- equivalently, the Generalised Axiom of Revealed Preference (GARP). Since this result is a very well-established one, a remark is warranted about what new value is derived from casting it in graph-theoretic terms. Invoking graph theory facilitates the development of ideas jointly related to both index number and revealed preference theory. Within the index number literature, graph theory is a well-established framework for analysing the structure and properties of multilateral indices where the matter of how individual countries' data are related in a network is a first-order research concern (for example, see  \textcite{hill2004constructing}). In contrast, the revealed preference literature seldom examines the identity of and relationship between individual observations within a dataset as a central object of interest.

This fact becomes more concrete when examining the actual content of Proposition \ref{prop:ref_consumer_characterisation}. First, in contrast to the individual prices and quantities typically used to state and test the theorem, Proposition \ref{prop:ref_consumer_characterisation} illustrates that GARP can be tested with alternative data: sufficient bilateral Laspeyres index numbers to construct $G$. Additionally, in terms of its input data, Proposition \ref{prop:ref_consumer_characterisation}'s rationalisation of a group of countries rather than a single one has few precedents in the revealed preference literature (a notable exception is \textcite{crawford_how_2013}) which when ported to the index number setting is the concept that allows for a welfare appraisal of multilateral formulae at the heart of this paper. Finally, it is very easy to see how Proposition \ref{prop:ref_consumer_characterisation} operates in light of the graphs involved. To this end, Figure \ref{fig:cewec_examples} presents some examples of $\mathcal{D}$ and their associated graphs to illustrate the connection between the data, the CEWEC, and the reference consumer.

\begin{figure}[htbp]
    \centering
    \captionsetup[subfigure]{justification=centering,font=small} 
    \centering

    \begin{subfigure}[t]{0.32\textwidth}
        \centering
        \resizebox{\linewidth}{!}{
        \begin{tikzpicture}[
            vertex/.style={circle, draw, fill=black, inner sep=0pt, minimum size=6pt},
            blue vertex/.style={vertex, draw=blue, fill=blue},
            cyan vertex/.style={vertex, draw=cyan, fill=cyan},
            edge/.style={-{Stealth[length=2.4mm]}, bend left=20},
            blue edge/.style={edge, draw=blue},
            weight/.style={font=\scriptsize, midway, fill=white, inner sep=1pt}
        ]
            \node[vertex] (a) at (0,0)  [label=left:$V_{1}$] {};
            \node[vertex] (b) at (3,0)  [label=right:$V_{2}$] {};
            \node[vertex] (c) at (1.5,2.6) [label=above:$V_{3}$] {};
            \draw[edge]      (a) to node[weight] {$33.67$} (b);
            \draw[edge]      (b) to node[weight] {$0.51$}  (a);
            \draw[edge]      (a) to node[weight] {$336.67$}(c);
            \draw[edge]      (c) to node[weight] {$0.10$}  (a);
            \draw[edge]      (b) to node[weight] {$500.05$}(c);
            \draw[edge]      (c) to node[weight] {$5.00$}  (b);
        \end{tikzpicture}}%
        \caption{$\mathbf{p}_{1}\!\equiv\!(1,1),\mathbf{p}_{2}\!\equiv\!(10,0.1),\mathbf{p}_{3}\!\equiv\!(0.1,10)$, 
                 $\mathbf{q}_{1}\!\equiv\!(1,2),\mathbf{q}_{2}\!\equiv\!(1,100),\mathbf{q}_{3}\!\equiv\!(1000,10)$}
    \end{subfigure}
    \hfill
    \begin{subfigure}[t]{0.32\textwidth}
        \centering
        \resizebox{\linewidth}{!}{%
        \begin{tikzpicture}[
            vertex/.style={circle, draw, fill=black, inner sep=0pt, minimum size=6pt},
            red vertex/.style={vertex, draw=red, fill=red},
            cyan vertex/.style={vertex, draw=cyan, fill=cyan},
            edge/.style={-{Stealth[length=2.4mm]}, bend left=20},
            red edge/.style={edge, draw=red},
            weight/.style={font=\scriptsize, midway, fill=white, inner sep=1pt}
        ]
            \node[red vertex] (a) at (0,0)  [label=left:$V_{1}$] {};
            \node[red vertex] (b) at (3,0)  [label=right:$V_{2}$] {};
            \node[red vertex] (c) at (1.5,2.6)[label=above:$V_{3}$] {};
            \draw[edge]      (a) to node[weight] {$1.05$} (b);
            \draw[red edge]  (b) to node[weight] {$0.88$}(a);
            \draw[red edge]  (a) to node[weight] {$0.95$}(c);
            \draw[edge]      (c) to node[weight] {$1.04$}(a);
            \draw[edge]      (b) to node[weight] {$1.48$}(c);
            \draw[red edge]  (c) to node[weight] {$0.96$}(b);
        \end{tikzpicture}}%
        \caption{$\mathbf{p}_{1}\!\equiv\!(2.5,4.5,2),\mathbf{p}_{2}\!\equiv\!(3.5,1,5.5),\mathbf{p}_{3}\!\equiv\!(5.5,3,2.5)$, 
                 $\mathbf{q}_{1}\!\equiv\!(5,3.5,1),\mathbf{q}_{2}\!\equiv\!(3.5,4,2.5),\mathbf{q}_{3}\!\equiv\!(3.5,2,5.5)$}
    \end{subfigure}
    \hfill
    \begin{subfigure}[t]{0.32\textwidth}
        \centering
        \resizebox{\linewidth}{!}{%
        \begin{tikzpicture}[
            vertex/.style={circle, draw, fill=black, inner sep=0pt, minimum size=6pt},
            blue vertex/.style={vertex, draw=blue, fill=blue},
            cyan vertex/.style={vertex, draw=cyan, fill=cyan},
            edge/.style={-{Stealth[length=2.4mm]}, bend left=20},
            blue edge/.style={edge, draw=blue, thick},
            cyan edge/.style={edge, draw=cyan, thick},
            red edge/.style={edge, draw=red},
            weight/.style={font=\scriptsize, midway, fill=white, inner sep=0pt}
        ]
            \node[blue vertex] (a) at (0,0) [label=left:$V_{1}$] {};
            \node[blue vertex] (b) at (3,0) [label=right:$V_{2}$] {};
            \node[cyan vertex] (c) at (3,3) [label=right:$V_{3}$] {};
            \node[cyan vertex] (d) at (0,3) [label=left:$V_{4}$] {};
            \draw[blue edge] (a) to node[weight] {$2.00$} (b);
            \draw[blue edge] (b) to node[weight] {$0.50$} (a);
            \draw[edge]      (a) to node[weight] {$0.20$} (c);
            \draw[edge]      (c) to node[weight] {$5.00$} (a);
            \draw[red edge]  (a) to node[weight] {$0.80$} (d);
            \draw[red edge]  (d) to node[weight] {$0.97$} (a);
            \draw[edge]      (b) to node[weight] {$0.10$} (c);
            \draw[edge]      (c) to node[weight] {$10.00$}(b);
            \draw[edge]      (b) to node[weight] {$0.47$} (d);
            \draw[edge]      (d) to node[weight] {$1.94$} (b);
            \draw[cyan edge] (c) to node[weight] {$3.33$} (d);
            \draw[cyan edge] (d) to node[weight] {$0.19$} (c);
        \end{tikzpicture}}%
        \caption{$\mathbf{p}_{1}\!\equiv\!(10,10),\mathbf{p}_{2}\!\equiv\!(1,2),\mathbf{p}_{3}\!\equiv\!(10,5),\mathbf{p}_{4}\!\equiv\!(4,15),\;
                 \mathbf{q}_{1}\!\equiv\!(5,5),\mathbf{q}_{2}\!\equiv\!(10,10),\mathbf{q}_{3}\!\equiv\!(1,1),\mathbf{q}_{4}\!\equiv\!(2,6)$}
    \end{subfigure}

    \caption{Examples of $\mathcal{D}$ and $G$. (a) The three-country numerical example from \textcite{diewert_axiomatic_1999}; here, a reference consumer exists. (b) A three-country case with $K=3$ where a reference consumer does not exist. The cycle in $G$ violating the CEWEC involves all three countries and is highlighted in \textcolor{red}{red}. (c) An example with $N=4,K=2$ where a reference consumer exists for two distinct country pairs (coloured \textcolor{cyan}{cyan} and \textcolor{blue}{blue}), but where the CEWEC is nonetheless violated for $G$ on the whole.} 
    \label{fig:cewec_examples}
\end{figure}

Before developing the consequences Proposition \ref{prop:ref_consumer_characterisation}, several points of clarification about the interpretation of the reference consumer are warranted. First, Proposition \ref{prop:ref_consumer_characterisation} does not say that the countries in the data actually possess the rationalising preference either individually or as a group. Rather, the result permits treating the data  \textit{as if} the countries in question behaved in accordance with these tastes without loss of generality. This is a necessary condition for interpreting multilateral indices as cost-of-living indices from some common taste, but is not sufficient for such an interpretation to be normatively compelling. This raises the question of welfare representativeness: just because we can treat multiple countries as if their consumption arose from a single set of tastes, is it meaningful to do so in order to make welfare statements? The assumption I make in this paper is that such an exercise can indeed be useful. Although this stance cannot be tested within the paper itself, it is worth noting that essentially the entirety of the literature about cost-of-living-based multilateral indices and cross-country welfare measures (including modern studies such as \textcite{jones2016beyond}, \textcite{cavallo2023product}, \textcite{argente2023measuring}, among many others) assume a global preference so are subject to this and related critiques. 

Second, the international reference consumer should not be misconstrued with a \textit{representative} agent where sub-aggregate utility functions are required to satisfy the relevant aggregation conditions (e.g. \cite{gorman_community_1953}). This is a more restrictive concept than the reference consumer's tastes which do not purport to represent an aggregate of countries. 

Finally, as illustrated by panel (c) in Figure \ref{fig:cewec_examples}, there is no guarantee that the condition in Proposition \ref{prop:ref_consumer_characterisation} will hold for a given dataset. This fact is what permits an appraisal of the taste assumptions underlying existing multilateral indices, but comes with the shortcoming that a dataset may not be compatible in its entirety with a reference consumer (and thus a multilateral cost-of-living index). This possibility may appear to create problems for constructing a multilateral index. Section \ref{subsec:fat_star_og} explains why this is not the case.

\subsection{Multilateral bounds on the reference consumer's cost of living} \label{subsec:mult_bounds}

I now derive multilateral cost-of-living bounds that must hold for any reference consumer. Some preliminary concepts are again required to articulate this idea.

\paragraph{Ordered walks and country sets.} Suppose we have constructed a reference consumer from the data using the graph $G$. To obtain implications for the multilateral cost of living, first consider a descending-index walk $G$ may contain and the countries reachable by it.

\begin{definition}[Revealed preference walks]\label{def:rp_walk}
A walk $(v_0,\dots,v_k)$ in $G=(V,E,w)$ is an
\emph{RP‑walk} if
\begin{align}
    w(v_{j-1},v_j)\le 1 \quad\text{for every } j=1,\dots,k. \label{eq:rp_walk}
\end{align}
We write $uRv$ when there exists an RP‑walk from vertex $u$ to
vertex $v$. If at least one inequality in \eqref{eq:rp_walk} does not hold with equality the walk is a \emph{strict RP‑walk}, denoted $uPv$.
\end{definition}

\begin{definition}[Revealed preference country sets]\label{def:rp_sets}
Using the relation R from Definition \ref{def:rp_walk}, for a vertex $v\in V$ define
\begin{align}
  \operatorname{VRP}(v) &= \{\,u\in V : u R v\,\}, \label{eq:VRP}\\
  \operatorname{VRW}(v) &= \{\,u\in V : v R u\,\}. \label{eq:VRW}
\end{align}
\end{definition}

Definition \ref{def:rp_walk} describes the relative size of country $i$'s and $0$'s real output by a sequence of intermediate price valuations. The walks have a standard revealed preference interpretation. Since $w(\upsilon_{i},\upsilon_{j})\leq 1$ implies that country $j$'s demand is affordable at country $i$'s but was not chosen, by the Weak Axiom of Revealed Preference (WARP) $\mathbf{q}_{i}$ is revealed preferred to $\mathbf{q}_{j}$. The walks extend this idea beyond individual country pairs to longer multilateral paths.

Collecting the reachable countries to and from an arbitrary vertex $\upsilon_{0}$ then leads to the sets in Definition \ref{def:rp_sets}: \eqref{eq:VRP} and \eqref{eq:VRW} contain countries with real output (consumption) that is larger (revealed preferred) or smaller (revealed worse) than $\upsilon_{0}$. The sets are related to the `revealed preferred' and `revealed worse' sets introduced in \textcite{varian_nonparametric_1982} but collect discrete observations rather than defining continuous regions in goods space.

\paragraph{Multilateral bounds on the cost of living.} I can now state the paper's central conceptual tool. Proposition \ref{prop:nonparam_col_bounds} gives sharp non-parametric multilateral bounds on the reference consumer's cost of living and positions these bounds relative to their classical bilateral counterparts. The multilateral bounds are the inner inequalities in Proposition \ref{prop:nonparam_col_bounds}.

\begin{proposition}\label{prop:nonparam_col_bounds}
For data $\mathcal{D}$ satisfying Proposition \ref{prop:ref_consumer_characterisation}, for all $i,j\in\{1,\ldots,N\}$ we have
\begin{align} \label{eq:bound_improvements_laspeyres}       
        \min\left\{{\frac{p_{j}^{k}}{p_{i}^{k}}}\right\} &\leq \frac{M^{-}(\mathbf{p}_{j},\upsilon_{i})}{M^{-}(\mathbf{p}_{i},\upsilon_{i})} \leq \frac{e(\mathbf{p}_{j},u_{RC}(\mathbf{q}_{i}))}{e(\mathbf{p}_{i},u_{RC}(\mathbf{q}_{i}))} \leq \frac{M^{+}(\mathbf{p}_{j},\upsilon_{i})}{M^{+}(\mathbf{p}_{i},\upsilon_{i})} \leq \frac{\mathbf{p}_{j}.\mathbf{q}_{i}}{\mathbf{p}_{i}.\mathbf{q}_{i}}
\end{align}
and
\begin{align} \label{eq:bound_improvements_paasche}
    \frac{\mathbf{p}_{j}.\mathbf{q}_{j}}{\mathbf{p}_{i}.\mathbf{q}_{j}}&\leq \frac{M^{-}(\mathbf{p}_{j},\upsilon_{j})}{M^{-}(\mathbf{p}_{i},\upsilon_{j})} \leq \frac{e(\mathbf{p}_{j},u_{RC}(\mathbf{q}_{j}))}{e(\mathbf{p}_{i},u_{RC}(\mathbf{q}_{j}))} \leq \frac{M^{+}(\mathbf{p}_{j},\upsilon_{j})}{M^{+}(\mathbf{p}_{i},\upsilon_{j})} \leq \max\left\{{\frac{p_{j}^{k}}{p_{i}^{k}}}\right\}.
\end{align}
where
\begin{align}
    M^{-}(\mathbf{p},\upsilon) &= \inf_{\mathbf{q}} \mathbf{p}.\mathbf{q} \text{ such that } \mathbf{p}_{i}\mathbf{q}\geq \mathbf{p}_{i}\mathbf{q}_{i} \text{ for all } i \in VRW(\upsilon) \label{eq:M_minus} \\
    M^{+}(\mathbf{p},\upsilon) &= \min_{\mathbf{q}_{i}} \mathbf{p}.\mathbf{q}_{i} \text{ such that } i \in VRP(\upsilon) \label{eq:M_plus}
\end{align}
\end{proposition}
\noindent \textbf{Proof.} See \hyperref[proof:nonparam_col_bounds]{\textbf{Appendix}}.

A few bibliographic remarks about Proposition \ref{prop:nonparam_col_bounds} are in order. The outermost objects in \eqref{eq:bound_improvements_laspeyres} and \eqref{eq:bound_improvements_paasche} are the classical bilateral bounds due to \textcite{konus_problem_1939} and \textcite{Pollak1971a}. The innermost bounds in \eqref{eq:bound_improvements_laspeyres} and \eqref{eq:bound_improvements_paasche} are non-parametric index numbers arising from revealed preference. The building blocks of these numbers, $M^{-}$ and $M^{+}$, are sharp bounds on the expenditure function which are well-known in the revealed preference literature.

Despite its adjacency to known results, to the best of my knowledge Proposition \ref{prop:nonparam_col_bounds} is new. There is some evidence that the multilateral bounds themselves are known: for example, the result is alluded to in \textcite{varian_nonparametric_1982} regarding an empirical application involving an object similar to the cost of living, and a special case of Proposition \ref{prop:nonparam_col_bounds} is articulated for homothetic tastes by \textcite{manser_analysis_1988} and \textcite{dowrick_true_1997}. However, a formal statement about the multilateral cost-of-living bounds for non-homothetic tastes and their position relative to their classical analogues surprisingly does not appear in the literature.

The multilateral bounds are easily computed by solving the small linear programs in \eqref{eq:M_minus} and \eqref{eq:M_plus}. Another property of the bounds is that they tighten as more data are used to construct the reference consumer's tastes. Since this property is the fundamental reason why the bounds can improve upon their bilateral counterparts I record this fact below.
\begin{fact}\label{fact:bounds_tighten_w_n}
    Suppose a graph $G$ associated with the data $\mathcal{D}$ satisfies the CEWEC. Let $G'$ be the graph associated with the data $\mathcal{D}'$ where $\mathcal{D}'$ comprises $\mathcal{D}$ with at least one additional observation. \bigskip

    \noindent Let $M^{-}_{G}(\mathbf{p},\upsilon)$, $M^{-}_{G'}(\mathbf{p},\upsilon)$, $M^{+}_{G}(\mathbf{p},\upsilon)$, and $M^{+}_{G'}(\mathbf{p},\upsilon)$ denote the bounds from Proposition \ref{prop:nonparam_col_bounds} associated with these graphs. If $G'$ satisfies the CEWEC, then
    \begin{align}
        M^{-}_{G'}(\mathbf{p},\upsilon)&\geq M^{-}_{G}(\mathbf{p},\upsilon) \label{eq:more_obs1} \\
        M^{+}_{G'}(\mathbf{p},\upsilon)&\leq M^{+}_{G}(\mathbf{p},\upsilon) \label{eq:more_obs2}
    \end{align}
\end{fact}
\noindent \textbf{Proof}. See \hyperref[proof:bounds_tighten_w_n]{\textbf{Appendix}}.

\paragraph{Index appraisal.} Appraising  existing cross-country comparison methods via Proposition \ref{prop:nonparam_col_bounds} is one of the paper's main contributions. Since any index number outside the interval (to use \eqref{eq:bound_improvements_laspeyres} as an example) $\left[\frac{M^{-}(\mathbf{p}_{j},\upsilon_{i})}{M^{-}(\mathbf{p}_{i},\upsilon_{i})},\frac{M^{+}(\mathbf{p}_{j},\upsilon_{i})}{M^{+}(\mathbf{p}_{i},\upsilon_{i})}\right]$ is incompatible with a cost of living, to do this, I simply record the frequency and extent to which such cases occur among the $N^{2}$ possible country pairs by counting bound violations.
\begin{definition}[Error rates and magnitudes] \label{def:econ_error_rate}
    Let $\mathcal{M}^{+}_{ij}$ and $\mathcal{M}^{-}_{ij}$  respectively denote $\frac{M^{-}(\mathbf{p}_{j},\upsilon_{i})}{M^{-}(\mathbf{p}_{i},\upsilon_{i})}$ and $\frac{M^{+}(\mathbf{p}_{j},\upsilon_{i})}{M^{+}(\mathbf{p}_{i},\upsilon_{i})}$ or their respective analogues in \eqref{eq:bound_improvements_paasche}. Consider an index $I$ with value $I_{ij}$ for country-pair $(i,j)$. $I$'s \emph{economic error rate} $\varepsilon_{I}$ is defined as the proportion of country pairs where the index $I$ violates one of the multilateral bounds, i.e.
    \begin{align}
    \varepsilon_{I} \equiv \frac{1}{N^2} \sum_{i=1}^{N} \sum_{j=1}^{N} \left( \mathbf{1}_{\{I_{ij} < \mathcal{M}_{ij}^{-}\}} + \mathbf{1}_{\{I_{ij} > \mathcal{M}_{ij}^{+}\}} \right)
    \label{eq:econ_error_rate}
    \end{align}
and we define $I$'s \emph{error magnitude} as the mean percentage value of bound violations:
    \begin{align}
    |\varepsilon|_{I} = \frac{\sum_{i=1}^{N} \sum_{j=1}^{N} \left( \mathbf{1}_{\{I_{ij} > \mathcal{M}_{ij}^{+}\}} \left( \frac{I_{ij}}{\mathcal{M}_{ij}^{+}} - 1 \right) + \mathbf{1}_{\{I_{ij} < \mathcal{M}_{ij}^{-}\}} \left( \frac{\mathcal{M}_{ij}^{-}}{I_{ij}}-1 \right) \right)}{\sum_{i=1}^{N} \sum_{j=1}^{N} \left( \mathbf{1}_{\{I_{ij} > \mathcal{M}_{ij}^{+}\}} + \mathbf{1}_{\{I_{ij} < \mathcal{M}_{ij}^{-}\}} \right)}
    \label{eq:avg_error_mag}
    \end{align}
\end{definition}

The main advantage of Definition \ref{def:econ_error_rate} is that it allows a comparison of multilateral indices' welfare validity that does not require an external index, assumed correct, to serve as the basis of such an appraisal. Definition \ref{def:econ_error_rate} jettisons this need because the data themselves are used to answer this question. The appraisal scheme can also be applied to comparison methods that involve inputs other than country-level price and quantity data (as in case of market exchange rates or theoretical aggregates with calibrated values).

As is the case for any appraisal method, index appraisal via Proposition \ref{prop:nonparam_col_bounds} is by no means free of limitations of its own. First, despite the reference consumer's cost-of-living bounds arising from the data, each of the bounds is nonetheless calculated with reference to the indifference base $u(\mathbf{q}_{i})$. This is an unavoidable consequence of the underlying tastes being non-homothetic which forces dependence of the cost-of-living on a base utility level. The implication of this fact for the index appraisal scheme is that any validity assessment is only as meaningful to the extent that the indifference base $u(\mathbf{q}_{i})$ is itself judged to be an appropriate comparison level for any given context, even when errors are averaged across country pairs as in \eqref{eq:econ_error_rate} and \eqref{eq:avg_error_mag}. An argument in favour of $u(\mathbf{q}_{i})$ is that it relates to the demand that is observed in country $i$. Yet admittedly, many other indifference bases could be chosen. This problem is a classic one in the index number literatures under non-homotheticities (see for example \textcite{pollak_theory_1989}) and is not avoided in this paper's framework (though neither is it avoided in any other non-homothetic comparison concept such as \textcite{neary_rationalizing_2004}). Another limitation in appraising indices via Proposition \ref{prop:nonparam_col_bounds} is that the identity of the reference consumer set is not within the researcher's control. For example, one may wish to construct a reference consumer for all countries in Sub-Saharan Africa but the data may not permit this. While techniques exist that could coerce additional observations into the reference consumer set (e.g. the Critical Cost Efficiency idea of \textcite{afriat1973system}), such extensions would alter the appraisal in ways beyond the scope of this paper to discuss.

\subsection{Index construction and a generalised star system} \label{subsec:fat_star_og}

In this section I define a novel multilateral index which I refer to as a `generalised star system' (GSS) constructed from the reference consumer's cost-of-living bounds. The intuition for this scheme is that the countries in a reference consumer set are assigned the geometric mean of their cost-of-living bounds. These countries jointly act as the centre of a star to value those outside of it, which are assigned the bounds and index values they would have \textit{if} they possessed the tastes of the reference consumer.

To further articulate this idea, first observe that because any point in the interval $\left[\mathcal{M}^{-}_{ij},\mathcal{M}^{+}_{ij}\right]$ is a valid index number, the reference consumer's cost-of-living bounds define a family of non-homothetic multilateral price indices that are welfare-interpretable and transitive for any fixed indifference base $u(\mathbf{q})$. Completing the construction of an index from this interval requires tackling two issues: (1) choosing a point among those permitted and (2) addressing how countries outside the reference consumer set are treated.

For the index value, I select the bounds' geometric mean $\sqrt{\mathcal{M}^{-}_{ij}\mathcal{M}^{+}_{ij}}$. This is a natural candidate relative to the existing literature because for a given reference consumer it can be interpreted as the non-homothetic analogue to the Ideal Afriat Index of \textcite{dowrick_true_1997} and collapses to the Fisher index under homothetic tastes when $N=2$. A proof of this fact is in the Supplemental Appendix. Addressing the second issue about extending the index to countries outside the reference consumer set is less straightforward. Potential trouble stems from the fact that adding a new country to a given reference consumer set will either produce a new, different reference consumer or violate CEWEC. To solve this problem, I construct a star system that is `generalised' in the sense that the \textit{multiple} countries from which $u(\mathbf{q})$ is constructed jointly act as the star's hub by   defining the reference consumer's tastes which are used to value other observations. Namely, any country $k$ lying outside of the reference consumer set has its cost-of-living index bounds derived from the feasible demands under its observed budget $(p_{k},m_{k})$ that are permitted under the reference consumer's tastes. These demands are easily found using the graph $G$ as follows. For the lower bound, we add a new vertex $\upsilon_{k}$ to $G$ along with directed edges pointing to $\upsilon_{k}$ and then calculate the expenditure function bound $M^{-}(\mathbf{p},\upsilon_{k})$ from \eqref{eq:M_minus}. An analogous procedure is not feasible for the upper bound in \eqref{eq:M_plus} because the observation $(p_{k},m_{k})$ may be disallowed by Proposition \ref{prop:ref_consumer_characterisation}. This problem does not occur for $M^{-}(\mathbf{p},\upsilon_{k})$ because the edges pointing to $\upsilon_{k}$ are always constructed with known quantities: $\mathbf{q}_{k}$ and prices from reference consumer set countries. I hence calculate the upper bound as
\begin{align}
    M^{+}(\mathbf{p},\upsilon_{k}) =& \sup_{\mathbf{q}} \mathbf{p}.\mathbf{q} \text{ such that } \label{eq:M_plus_outside_star} \\ 
    & \mathbf{p}_{i}\mathbf{q}\geq \mathbf{p}_{i}\mathbf{q}_{i} \text{ for all } i \in VRW(\upsilon_{k}) \text{ and } \nonumber \\
    & \mathbf{p}_{k}\mathbf{q}=m_{k} \nonumber
\end{align}
which clearly bounds $e(\mathbf{p},u_{RC}(\mathbf{q}_{k}))$ as the chosen $\mathbf{q}$ must lie on the budget $(\mathbf{p},m_{k})$ by the final constraint. \eqref{eq:M_plus_outside_star} is the largest expenditure among bundles on $(\mathbf{p},m_{k})$ permitted by the revealed preference restrictions arising from the reference consumer's tastes. With both bounds calculated, a geometric mean is then taken to obtain an index value. 

A final index requirement applies when there are multiple non-reference consumer countries. When there are two or more non-reference consumer countries which require index values, the resulting index values may be mutually inconsistent in the sense of \textcite{adams2020mutually} because solving the linear program in \eqref{eq:M_plus_outside_star} requires forecasting demand that when done more than once, may generate bundles that jointly violate Proposition \ref{prop:ref_consumer_characterisation}. This problem can be readily solved by introducing additional constraints and solving the associated mixed integer linear program developed in \textcite{adams2020mutually}. Importantly, this issue never arises for countries exclusively within the reference consumer set as a result of Proposition \ref{prop:nonparam_col_bounds}. Within this set, the bounds are mutually consistent by construction\footnote{To see this, observe that any case of mutual inconsistency must either contradict the fact that the data violate Proposition \ref{prop:ref_consumer_characterisation} for the reference consumer set, or the fact that the bounds in Proposition \ref{prop:nonparam_col_bounds} are sharp. A simple numerical illustration of this point for three countries is also provided in the case of homothetic tastes by \textcite[p.~45]{dowrick_true_1997}.} (which is where the source of their improvements over the classical bounds arises from) so index values that lie between them can always be calculated jointly while never violating rationality.

Figure \ref{fig:minimal_example_2} illustrates the generalised star system by continuing the numerical example from Section \ref{subsec:motivating_example}. Under homotheticity, this scheme extends the validity of the Afriat Ideal Index to countries whose data would otherwise preclude a rationalising preference.  

\begin{figure}[h!]
    \centering
    \captionsetup[subfigure]{justification=centering,font=small}

    \begin{subfigure}[t]{0.32\textwidth}
        \centering
        \resizebox{\linewidth}{!}{%
        \begin{tikzpicture}[
                vertex/.style={circle, draw, fill=black, inner sep=0pt, minimum size=6pt},
                red vertex/.style={vertex, draw=red, fill=red},
                cyan vertex/.style={vertex, draw=cyan, fill=cyan},
                edge/.style={-{Stealth[length=2.4mm]}, bend left=20},
                red edge/.style={edge, draw=red},
                weight/.style={font=\scriptsize, midway, fill=white, inner sep=1pt}
            ]
                \node[vertex] (a) at (0,0)        [label=left:$A$] {};
                \node[vertex] (b) at (3,0)        [label=right:$B$] {};
                \node[vertex] (c) at (1.5,2.6)    [label=above:$C$] {};
                \node[red vertex] (d) at (-1.5,-2.6) [label=below:$D$] {};

                \draw[edge] (a) to node[weight] {$1.330$} (b);
                \draw[edge] (b) to node[weight] {$0.823$}  (a);
                \draw[edge] (a) to node[weight] {$0.915$}(c);
                \draw[edge] (c) to node[weight] {$1.400$}  (a);
                \draw[edge] (b) to node[weight] {$0.588$}(c);
                \draw[edge] (c) to node[weight] {$1.700$}  (b);
                \draw[red edge] (a) to node[weight] {$0.723$}(d);
                \draw[red edge] (d) to node[weight] {$0.963$}  (a);
        \end{tikzpicture}}%
        \caption{Violation of CEWEC}
    \end{subfigure}
    \hfill
    \begin{subfigure}[t]{0.32\textwidth}
        \centering
        \resizebox{\linewidth}{!}{%
        \begin{tikzpicture}[
                vertex/.style={circle, draw, fill=black, inner sep=0pt, minimum size=6pt},
                green vertex/.style={vertex, draw=ForestGreen, fill=ForestGreen},
                blue vertex/.style={vertex, draw=blue, fill=blue},
                edge/.style={-{Stealth[length=2.4mm]}, bend left=20},
                red edge/.style={edge, draw=red},
                weight/.style={font=\scriptsize, midway, fill=white, inner sep=1pt}
            ]
                \node[green vertex] (a) at (0,0)        [label=left:$A$] {};
                \node[green vertex] (b) at (3,0)        [label=right:$B$] {};
                \node[green vertex] (c) at (1.5,2.6)    [label=above:$C$] {};
                \node[blue vertex]  (d) at (-1.5,-2.6)  [label=below:$D$] {};

                \draw[gray, dashed] (1.5,0.8667) circle [radius=1.732];

                \draw[edge] (a) to node[weight] {$1.330$} (b);
                \draw[edge] (b) to node[weight] {$0.823$}  (a);
                \draw[edge] (a) to node[weight] {$0.915$}(c);
                \draw[edge] (c) to node[weight] {$1.400$}  (a);
                \draw[edge] (b) to node[weight] {$0.588$}(c);
                \draw[edge] (c) to node[weight] {$1.700$}  (b);
                \draw[edge] (a) to node[weight] {$0.723$}(d);
        \end{tikzpicture}}%
        \caption{Generalised star}
    \end{subfigure}
    \hfill
    \begin{subfigure}[t]{0.32\textwidth}
        \centering
        \resizebox{\linewidth}{!}{%
        \begin{tikzpicture}[
                vertex/.style={circle, draw, fill=black, inner sep=0pt, minimum size=6pt},
                green vertex/.style={vertex, draw=ForestGreen, fill=ForestGreen},
                blue vertex/.style={vertex, draw=blue, fill=blue},
                edge/.style={-{Stealth[length=2.4mm]}, bend left=20},
                red edge/.style={edge, draw=red},
                weight/.style={font=\scriptsize, midway, fill=white, inner sep=1pt}
            ]
                \node[green vertex] (a) at (0,0)        [label=left:$A$] {};
                \node[blue vertex]  (b) at (3,0)        [label=right:$B$] {};
                \node[blue vertex]  (c) at (1.5,2.6)    [label=above:$C$] {};
                \node[blue vertex]  (d) at (-1.5,-2.6)  [label=below:$D$] {};

                \draw[edge] (a) to node[weight] {$1.330$} (b);
                \draw[edge] (a) to node[weight] {$0.915$}(c);
                \draw[edge] (a) to node[weight] {$0.723$}(d);
        \end{tikzpicture}}%
        \caption{Laspeyres star}
    \end{subfigure}

    \caption{\textbf{A generalised star system.} Consider using the reference consumer underlying the example data in Figure \ref{fig:minimal_example_1} to find $P_{A}/P_{D}$ for a new fourth country with data $ \mathbf{p}_{D}\!\equiv\!(10,4),\; \mathbf{q}_{D}\!\equiv\!(10,2)$. A reference consumer does not exist for \{A,B,C,D\} so Proposition \ref{prop:nonparam_col_bounds} does not apply (panel (a)). Panel (b) solves this issue with a generalised star system to calculate the index value $\sqrt{\mathcal{M}^{-}_{AD}\mathcal{M}^{+}_{AD}}$ without altering the original reference consumer set. Further details of the calculation for $\mathcal{M}^{+}_{AD}$ given by \eqref{eq:M_plus_outside_star} are as follows: the two countries in the set $VRW(\upsilon_{D})$ are $A$ and $C$; the demanded bundle for $D$ found by the linear program under the reference consumer's tastes is $\mathbf{q}*=(0,27)$; and the resulting index value is $\sqrt{0.87(2.25)}=1.399$. I call the scheme a `generalised star system' because multiple countries (panel (b), $\{A,B,C\}$) serve as the star's hub as opposed to just one (panel (c), where the centre of the star is country $\{A\}$). 
    \label{fig:minimal_example_2}}
\end{figure}

\section{Main applications} \label{sec:application}

In this section I apply the ideas in Section \ref{sec:intlrefcons} to identify an international reference consumer, appraise the welfare validity of several major multilateral indices, and re-visit two uses of these indices: the size of the world and inequality within it. Section \ref{subsec:data} describes the data. Section \ref{subsec:constructing_ref_consumer} presents the reference consumer and their tightened bounds. Section \ref{subsec:results} conducts an index appraisal using this reference consumer set. Finally, in Section \ref{subsec:output_and_inequality}, I turn to  major uses of international comparisons and re-calculate global output and inequality using the generalised star system from Section \ref{subsec:fat_star_og}. A Supplemental Appendix provides further notes on the data, repeats these exercises under homotheticity, and discusses how to correct welfare-inconsistent indices.

\subsection{Data} \label{subsec:data}

My main source of data is the United Nations International Comparison Program (ICP). The ICP collects expenditures in national currency units from participating economies and computes purchasing power parities (PPPs) in order to convert these expenditures into a common currency. The data comprise relative prices in the form of bilateral PPPs $\pi^{k}_{ib}\equiv p_{i}^{k}/p_{b}^{k}$ (with the base country $b$ being the United States) and per-capita nominal expenditures $e^{k}_{i}$ as defined in Section \ref{sec:notation} for the ICP's 2017 reference year. These prices and expenditures are reported at the ‘basic heading’ level, which are the finest GDP subclasses for which expenditure data are provided to the ICP by national statistical authorities. Despite this granularity, since basic headings are composites of individual goods their standalone units do not have a meaningful physical interpretation. My application dataset comprises $N=173$ countries and $K=114$ basic headings for which prices and demands are available for each country, and which comprise a standard consumption aggregate known as Household and NPISHs' Final Consumption Expenditure. Alongside the ICP data I also use market exchange rates from the International Monetary Fund's Exchange Rates dataset. These rates are annual averages over 2017 expressed in units of each country's domestic currency per US dollar. Further notes about the data are in the Supplemental Appendix.

\subsection{The reference consumer and their cost-of-living bounds} \label{subsec:constructing_ref_consumer}

\paragraph{A near-global reference consumer for the 2017 ICP.} My first finding is that there exists a reference consumer whose tastes account for all but one of the world's countries and over 99\% of output in the 2017 sample (valued using PPPs relative to the United States). These particular tastes are distinguished in the sense that they span the maximum number of countries jointly  consistent with Proposition \ref{prop:ref_consumer_characterisation}. The countries are shown in Figure \ref{fig:max_clique_irc}.

\begin{figure}[h!]
    \centering
    \includegraphics[width=0.8\textwidth]{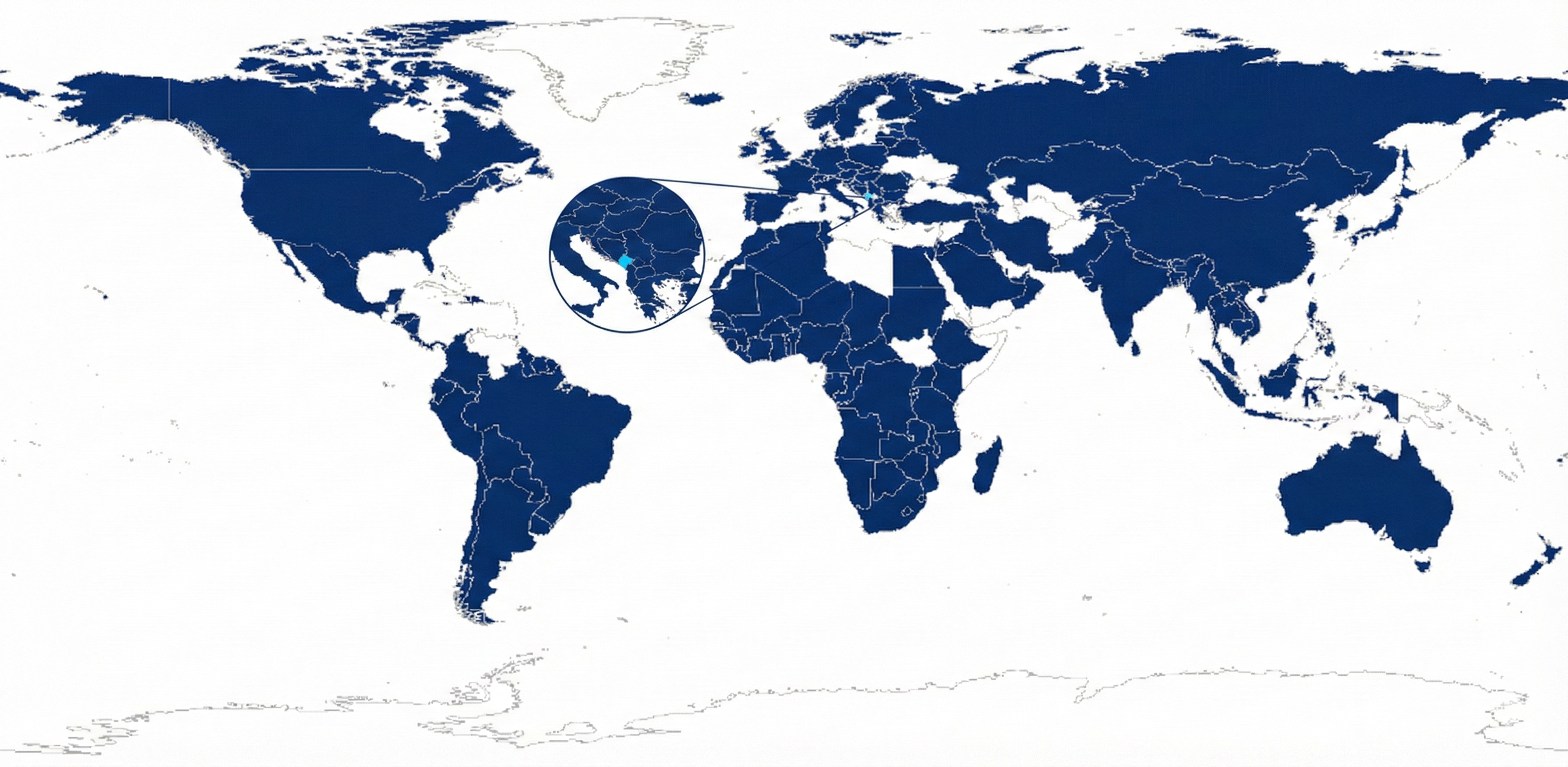}
    \caption{Reference consumer countries (dark blue). The only incompatible country (Montenegro) is in light blue. White countries are not represented in the data. \textbf{Note}: A natural question is how `far' MNE's data are from being consistent with the reference consumer's tastes. The difference appears small. For example, the Money Pump Index (MPI) of \textcite{echenique2011money} that measures the degree of revealed preference violations, and which in this context can be interpreted as a potential arbitrageur's profit from a failure of the Law of One Price, attains a value of just 1.34\% of country-pair expenditure.}
    \label{fig:max_clique_irc}
\end{figure}

The large number of countries included in the reference consumer set provides a broad base to appraise indices with the bounds in Proposition \ref{prop:nonparam_col_bounds}. This base also lends credence to multilateral indices that assume common tastes across countries. However, it is worth reminding the reader that Figure \ref{fig:max_clique_irc} should not be interpreted as evidence that these countries have identical preferences, but rather that the world's observed demands may be treated \textit{as if} they were generated from a single set of tastes.

\paragraph{Improvements to the classical bilateral cost-of-living bounds.} Consistent with Proposition \ref{prop:nonparam_col_bounds}, the multilateral cost-of-living bounds are tighter than their bilateral classical counterparts. To illustrate these improvements, an example is shown below at Table \ref{tab:bounds_comparison_usa_chn} for the PPP exchange rates for country pair example of the United States and China, which is the world's largest bilateral trading relationship. In this case the width of the non-parametric multilateral bounds is 17 percentage points tighter than the classical bound, with the improvement for this country pair being driven entirely by the lower bound. 
\begin{table}[ht]
    \centering
    \begin{tabular}{lccc ccc|c}
        \toprule
        \toprule
        & \multicolumn{3}{c}{Classical bilateral bounds} 
        & \multicolumn{3}{c}{Multilateral bounds} 
        & \\
        \cmidrule(lr){2-4} \cmidrule(lr){5-7}
         & Min.\ price relative & Laspeyres & $\Delta_{C}$ 
         & $\mathcal{M}^{-}$ & $\mathcal{M}^{+}$ & $\Delta_{M}$ 
         & (1-$\Delta_{M}/\Delta_{C}$) \\
        \midrule
        USA--China 
         & 0.629 & 4.39 & 3.76 
         & 1.28 & 4.39 & 3.11  
         & 0.173 \\
        \bottomrule
        \bottomrule
    \end{tabular}
    \caption{Values of the classical bilateral and multilateral cost-of-living bounds based on reference consumer tastes for the United States and China. Each reported bound is calculated for the indifference base $u(\mathbf{q}_{US})$ along with the width of the classical and multilateral bounds ($\Delta_{C}$ and $\Delta_{M}$) and the proportional difference in bound widths ($1-\Delta_{M}/\Delta_{C}$). The 2017 USD/CNY market exchange rate for this currency pair was 6.759, which lies well outside the  upper bounds and which implies an overvaluation of the yuan relative to the US dollar from the perspective of the reference consumer's cost of living.}
    \label{tab:bounds_comparison_usa_chn}
\end{table}

A similar pattern holds across all of the country pairs in the data and is shown in Figure \ref{fig:bound_improv_violin_GARP}. The average multilateral bound improvement relative to the classical bilateral bounds is 22.1\%. However, this average masks heterogeneity across the bounds and country pairs. As shown in Figure \ref{fig:bound_improv_violin_GARP}, the majority of improvements occur for the lower bound rather than upper bound. This is not in itself surprising as the classical lower bound is a highly conservative quantity that requires extreme behaviour on the part of the consumer to achieve. (Nonetheless, the remark of \textcite[p.~18]{pollak_theory_1989} that `it is not our lack of ingenuity but the inherent logic of the situation which prevents us from finding better [bounds]' applies here.) A less obvious fact is the distribution of improvements \textit{within} each bound type, which also differs markedly between the upper and lower bounds. Here, many countries attain only limited improvements upon their minimum price relatives which runs against the earlier intuition about the bound being a conservative one.

\begin{figure}[h!]
    \centering
    \includegraphics[width=0.8\textwidth]{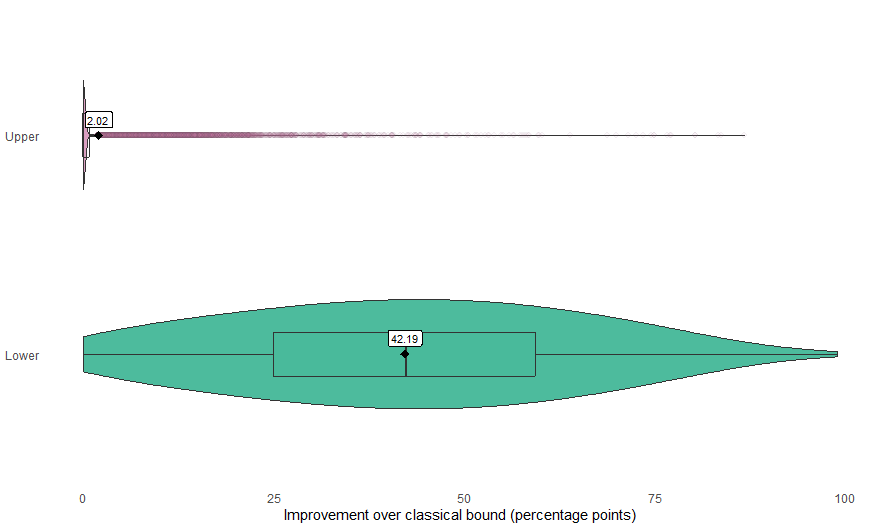}
    \caption{Improvements between the classical bilateral and non-parametric multilateral bounds based \eqref{eq:bound_improvements_laspeyres} for the $172^2 = 29,584$ country pairs in the reference consumer set. \textbf{Note}: bound improvements are defined as the percentage point difference between the widths of the upper and lower bounds between the bilateral and multilateral cases. The labeled points show the mean percentage point improvement for the lower and upper bounds respectively.}. 
    \label{fig:bound_improv_violin_GARP}
\end{figure}

Improvements to the classical cost-of-living bounds are also not uniformly distributed across space. As shown in Figure \ref{fig:tightened_bounds_map} for $\mathcal{M}^{-}_{ij}$, countries with the greatest improvements to their bounds are concentrated within the Middle East and South America. In contrast, many countries in Sub-Saharan Africa experience modest increases. These differences are a consequence of the structure of $G$ and its relationship to relative real incomes. To take the example of the lower bound patterns seen in Sub-Saharan African: country budgets in this region are unlikely to have edge weights less than one given their low per-capita incomes reduces the number of countries in the sets $VRW(\upsilon)$ and $VRP(\upsilon)$. However, this relationship is not a mechanical one as the greatest bound improvements do not always occur for the highest-income countries in the sample.
\begin{figure}[h!]
    \centering
    \includegraphics[width=0.8\textwidth]{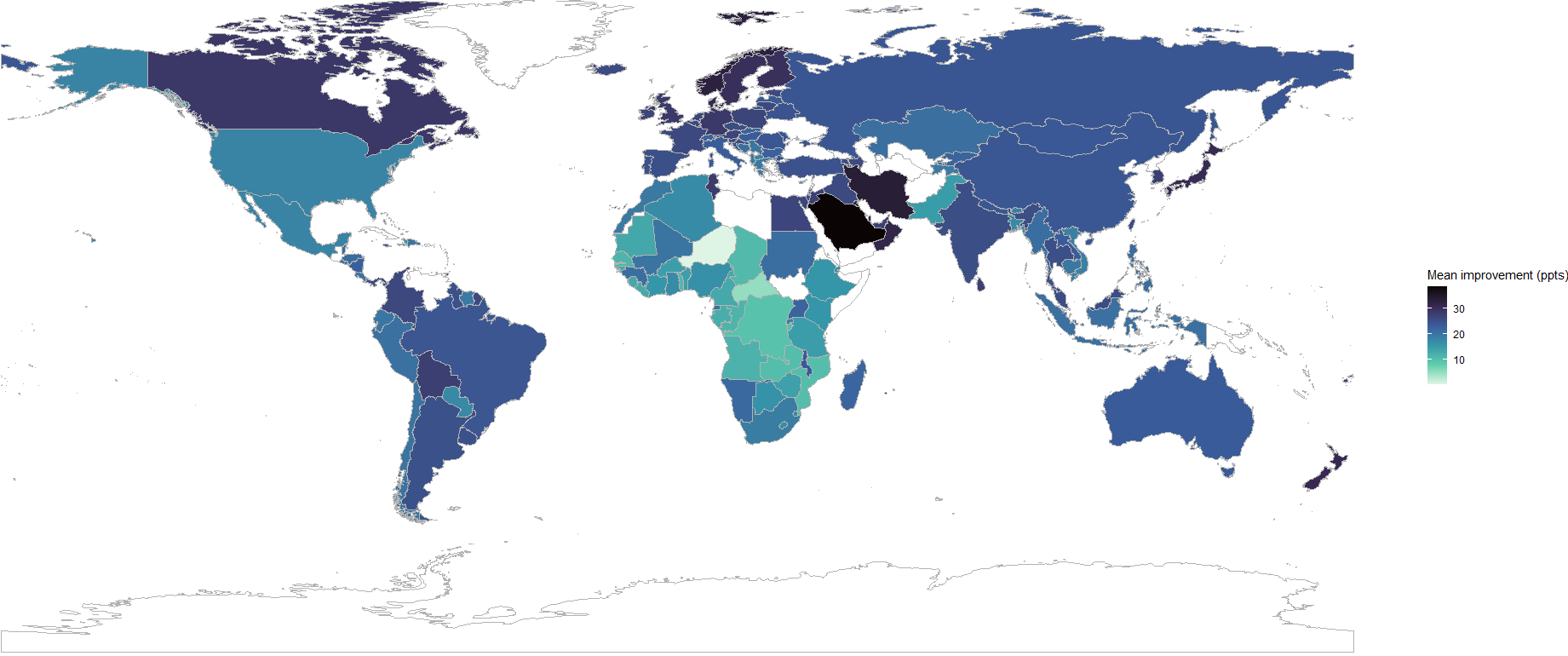}
    \caption{Mean percentage improvements from the non-parametric lower bound $\mathcal{M}^{-}_{ij}$ on the bilateral cost-of-living index relative to the classical lower bound $\min\left\{\frac{p_{j}^{k}}{p_{i}^{k}}\right\}$ under the reference consumer's tastes. The bound improvement for country $i$ is the sum of bilateral proportional difference averaged across all country-pairs using country $i$'s indifference level as a base, i.e. $\frac{1}{N} \sum_{j=1}^{N} (1-\mathcal{M}^{-}_{ij}/\min\{\{p_{j}^{k}/p_{i}^{k}\})$.}
    \label{fig:tightened_bounds_map}
\end{figure}

\subsection{Appraising multilateral indices' welfare validity} \label{subsec:results}

\paragraph{Index selection and research design rationale.} I now appraise several major multilateral indices, along with market exchange rates, in a scheme analogous to treatment arms in an experiment. The choice of indices appraised are: (1) the Geary-Khamis index (GK; \cite{geary_note_1958} and \cite{khamis_new_1972}); (2) the bilateral Fisher index; (3) the multilateral GEKS index (\cite{gini_quelques_1924}, \cite{elteto_problem_1964}, and \cite{szulc_indices_1964}); (4) the bilateral Törnqvist index; (5) and the multilateral CCD index of \textcite{caves_economic_1982}. A brief description of each of these indices is at in the Supplemental Appendix.

The choice of these indices lies beyond their broad application and influence in the index number and international comparisons literatures. Collectively, the five indices possess theoretical properties that differ from each other such that the bounds in Proposition \ref{prop:nonparam_col_bounds} and which allow me to quantify the welfare consistency of three index features that are often-discussed in the international comparisons literature. These are namely:
\begin{enumerate}
    \item the tension between \textit{additive} multilateral indices (that is, ones where aggregate real expenditure equals the sum of its components) which as shown by \textcite{diewert_axiomatic_1999}, cannot account for consumer substitution in actual quantities and \textit{superlative} ones, which account for substitution effects but which are inconsistent in aggregation;
    \item the implication assuming \textit{homothetic} preferences versus allowing for \textit{non-homothetic} tastes (e.g. those with income effects); and 
    \item the welfare effect of imposing the \textit{circularity} property (that is, that comparisons between any two countries are path-independent with $P_{ik}=P_{ij}P_{jk}$) upon bilateral indices to produce consistent multilateral comparisons.
\end{enumerate}
The logic of how the selected indexes achieve these comparisons is analogous to the selection of treatment conditions in an experiment. Table \ref{tab:compared_indices_properties} and Table \ref{tab:compared_indices_identification} summarise the relevant index properties and pairwise comparisons designed to shed light on these  areas. Comparisons of the relative consistency of additivity, homotheticity, and circularity, can be made by calculating either or both $\Delta \varepsilon$ or $\Delta |\varepsilon|$ for a suitable pair of indices. For example, comparing the economic error rates of either the Fisher vs. GEKS or Törnqvist vs. CCD indices isolates something akin to the `welfare cost of circularity' without confounding factors like additivity or homothetic/non-homothetic tastes assumptions being present. While comparisons have long been of interest and calculated in the literature between specific indices, it is difficult to draw conclusions about the extent of these comparisons' consistency with an economic welfare interpretation because in most cases at least one must be assumed to be correct. The bounds solve this problem since they are defined independently of any index formula and instead arise from the empirical content of the theory. 

Despite the advantages of this research design, it is important to note that the five appraised indices do not provide a general conclusion about the comparative welfare validity of the examined index properties. For example, it may possible that a circular, additive index that is not Geary-Khamis performs very differently to it. This limitation can again be interpreted analogously to that of an experiment.  Standalone experiments rarely possess sufficient external validity to settle a research question without complementary evidence, but this does not void their value; so too is this the case for the index appraisal exercise.

\begin{table}[h!]
\centering
\begin{tabular}{lccc}
\toprule
\toprule
\textbf{Method} & \textbf{Satisfies circularity? }&\textbf{ Additive?} & \textbf{Superlative?} \\
\midrule
Market exchange rates & \ding{51} & \ding{51} & \ding{55} \\
\midrule
Geary-Khamis & \ding{51} & \ding{51} & \ding{55} \\
\midrule
Fisher & \ding{55} & \ding{55} & \ding{51} (homothetic tastes) \\
GEKS & \ding{51} & \ding{55} & \ding{51} (homothetic tastes) \\ 
\midrule
Törnqvist & \ding{55} & \ding{55} & \ding{51} (non-homothetic tastes) \\
CCD & \ding{51} & \ding{55} & \ding{51} (non-homothetic tastes) \\ 
\bottomrule
\bottomrule
\end{tabular}
\caption[Caption for LOF]{Selected theoretical properties of the five compared indices and market exchange rates. \textbf{Note:} although it is not superlative, the Geary-Khamis index is exact for linear (additive) and Leontief preferences. (I thank Erwin Diewert for alerting me to this fact.)}
\label{tab:compared_indices_properties}
\end{table}
\begin{table}[h!]
\centering
\begin{tabular}{ll}
\toprule
\toprule
\textbf{Comparison} & \textbf{Identifying error rate differences} \\
\midrule
Additive vs. non-additive & (1a): $\varepsilon_{GK}-\varepsilon_{GEKS}$ and/or \\
& (1b): $\varepsilon_{GK}-\varepsilon_{CCD}$ \\
Homothetic vs. non-homothetic & (2): $\varepsilon_{GEKS}$ \\
Circular vs. non-circular & (3a): $\varepsilon_{GEKS} - \varepsilon_{Fisher}$ and/or \\
& (3b): $\varepsilon_{CCD} - \varepsilon_{T\ddot{o}rnqvist}$ \\
\bottomrule
\bottomrule
\end{tabular}
\caption{The pairwise comparisons and economic error rate differences that via the bounds $\mathcal{M}^{-}_{ij}$ and $\mathcal{M}^{+}_{ij}$ isolate the effect of imposing the conditions of additivity, homotheticity, or circularity upon a specific index number method on that index's consistency with a cost-of-living index based on the reference consumer's tastes.}
\label{tab:compared_indices_identification}
\end{table}

\paragraph{Index appraisal results.} Figure \ref{fig:index_error_chart_GARP} displays $\varepsilon$ and average error magnitudes $|\varepsilon|$ for the considered indices for the 2017 ICP reference year's data alongside these quantities for market exchange rates. The results for the pairwise comparisons in Table \ref{tab:compared_indices_identification} are in Table \ref{tab:compared_indices_results}. 

\begin{figure}[h!]
    \centering
    \includegraphics[width=0.8\textwidth]{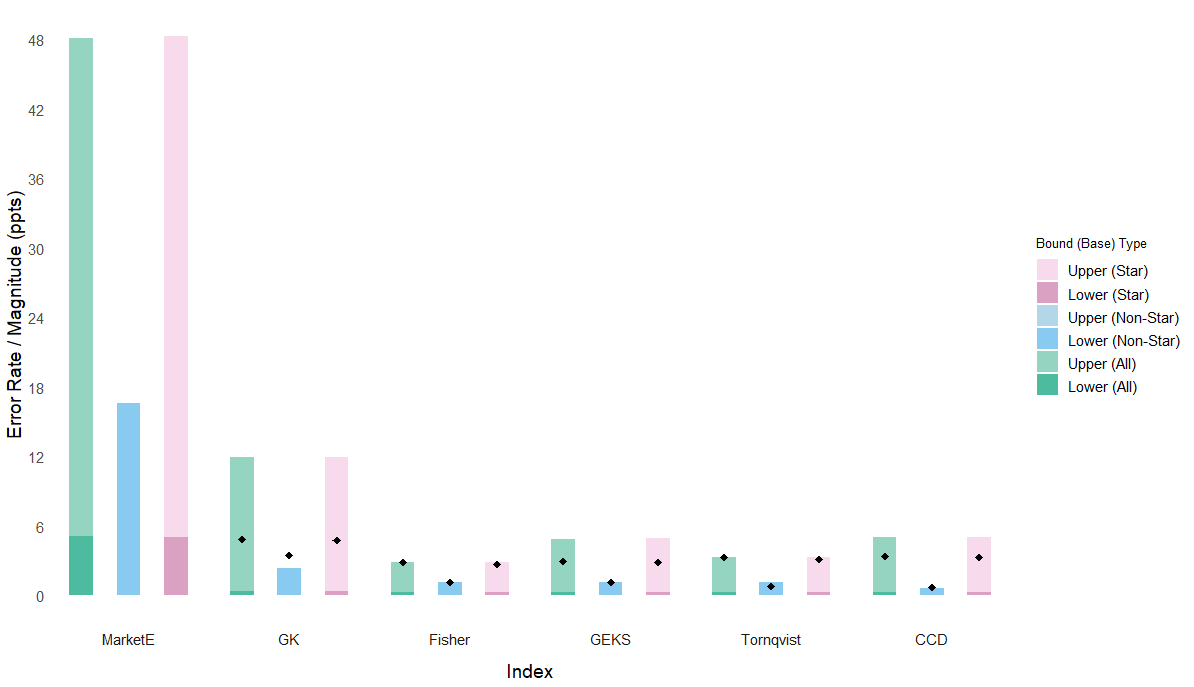}
    \caption{The average error rate $\varepsilon$ (coloured bars) and mean error magnitude $|\varepsilon|$ (black diamonds) for the five considered indices and market exchange rates based on the reference consumer tastes and \eqref{eq:bound_improvements_laspeyres}. \textbf{Note}: error rates are reported for both upper (lighter-coloured bars) and lower (darker-coloured bars) non-parametric bounds. Error rates and magnitudes are reported for three country-pair types: (1) all country pairs (green bars); (2) pairs where the base country $b$ (in the sense of $u(\mathbf{q}_{b}$) is a member of the reference consumer set (pink bars); and (3) pairs where this is not the case (blue bars). Error magnitudes for market exchange rates lie on average 7-10x outside the bounds so are omitted from the plot.} 
    \label{fig:index_error_chart_GARP}
\end{figure}
\begin{table}[h!]
\centering
\begin{tabular}{lcc|lcc}
\toprule
\toprule
\textbf{Index} & $\varepsilon$ & $|\varepsilon|$ & \textbf{Index pairs} & $\Delta\varepsilon$ & $\Delta|\varepsilon|$ \\
\midrule
$Fisher$ & 2.82 & 2.87 & (3a): $GEKS-Fisher$ & 2.04 & 0.05  \\
$T\ddot{o}rnqvist$ & 3.26 & 3.29 & (3b): $CCD-T\ddot{o}rnqvist$ & 1.74 & 0.07  \\
$GEKS$ & 4.86 & 2.92 & (1b): $GK-CCD$ & 6.90 & 1.49  \\
$CCD$ & 5.00 & 3.36 & (1a): $GK-GEKS$ & 7.04 & 1.93  \\
$GK$ & 11.90 & 4.85 &  \\
\bottomrule
\bottomrule
\end{tabular}
\caption{$\varepsilon$, $|\varepsilon|$, and relevant pairwise comparisons for the five considered indices.}
\label{tab:compared_indices_results}
\end{table}

Several patterns emerge among the error rates which apply to all of the appraised indices and rates. First, Figure \ref{fig:index_error_chart_GARP} strengthens the evidence against using market exchange rates for welfare applications: the rates' large errors compared to all the appraised indices suggest they are not a valid measure of economic welfare. Second, the majority of bound violations occur with respect to the upper bound rather than the lower one. An implication of this is that welfare inconsistencies generally lead to cross-country PPPs being over-estimated (and thus real incomes being under-estimated). However, the mean magnitude of this bias -- leading to the second pattern in Figure \ref{fig:index_error_chart_GARP} -- is within 3-5 percentage points of the relevant bound for the five indices. This is small when compared to the analogous error for market exchange rates, and is also small in practical terms when compared against known sources of bias in international comparisons for example, \textcite{deaton2017trying} document and investigate inter-wave differences in PPP estimates for ICP surveys between 2005 and 2011 that saw consumption PPPs for some regions overestimated by 18-26\%. When taken together with the fact that between c. 88-97\% of PPPs already lie within the allowed theoretical bounds for the five indices, Figure \ref{fig:index_error_chart_GARP} can be interpreted as evidence supporting the welfare validity of major multilateral index number methods used for international comparisons of prices and income, which despite their diverse theoretical foundations and origins, feature in significant and distinct applications.

The appraisal also reveals notable economic differences between each of the five appraised indices. First, relative to the other compared indices (each of which is superlative), errors are highest for the GK index relative to the superlative indices regardless of whether latter are bilateral or multilateral. This suggests the existence of substitution effects in the data and the fact that, as shown by \textcite{diewert_axiomatic_1999}, the GK index cannot incorporate these effects in the comparisons it generates. A second fact relates to $\varepsilon_{Fisher}$ and $\varepsilon_{GEKS}$: imposing the assumption of homothetic tastes in these indices does not appear to violate a cost-of-living interpretation for the majority of the country-pair comparisons in the data; in these cases 3-5\% of pairs lie outside one of the non-homothetic bounds. This is a surprising finding given the strength of the homotheticity assumption, but is consistent with analagous tests in the empirical revealed preference for within-country time series data (e.g. \cite{manser_analysis_1988}). A third distinction between the compared indices revealed by the bounds relates to circularity: imposing the `GEKS' or `CCD' step to the Fisher or Törnqvist to render the bilateral indices circular results in a 1.7-2 percentage point increase in the proportion of country pair comparisons that lie outside of the permitted welfare bounds. This finding suggests that imposing transitivity does not meaningfully affect welfare-validity. This finding is consistent with the existing numerical evidence about this topic in the literature. For example, \textcite{alterman_international_1999} find limited quantitative differences between the Fisher and GEKS indices. Finally, beyond individual index properties, the comparisons in Table \ref{tab:compared_indices_results} also permit an investigation of which of the three properties of additivity, circularity, and homotheticity are the most egregeious in terms of their impact on welfare-validity. In this regard, the ranking of these properties from `worst' to `best' is additivity, homotheticity, and circularity. 

\subsection{World output and inequality when accounting for income effects} \label{subsec:output_and_inequality}

\noindent As described in the Introduction, cross-country price and income comparisons underpin applications across many fields of economics. Given this, I now move beyond appraising the empirical welfare content of existing indices to instead ask: what do the ideas and results in this paper imply for the ends these comparisons are typically made for?

To this end, I use the generalised star system to re-calculate two fundamental aggregate statistics in cross-country comparisons: global output and inequality. I compare these to the current best practice ICP measurement which uses GEK parities. I find that when calculated using the generalised star system (which takes into account income effects due to the non-homotheticity of the underlying preference) with the United States' consumption bundle as an indifference base, the world in 2017 was substantially larger and more equal than the current standard measurements conducted by the ICP.

\paragraph{Parity differences with income effects.} Figures \ref{fig:geks_vs_gss_scatter} and \ref{fig:geks_vs_gss_bounds_position} show the root cause of my results in this sub-section: the differences between the PPPs generated from the generalised star system and the GEKS methods. The GSS parities are on average two-thirds of the size of their corresponding GEKS value and are smaller than the GEKS parities for almost all the countries in the sample. These differences correlate closely with per-capita expenditure: rich countries typically show the smaller changes between index values compared to poor ones suggesting that ICP/GEKS may function as if countries were valued using a rich country's consumption bundle irrespective of income. Moreover, as shown in Figure \ref{fig:geks_vs_gss_bounds_position}, for nearly all countries the GEKS values hew very closely to the upper non-parametric bound. As a result, when interpreted as the cost of living, the GSS values imply that welfare relative to the United States is generally higher than indicated by the conventional ICP measures, which has a cost of living interpretation under the hypothesis of homotheticity. Intriguingly, this empirical finding contrasts with that of the trade and product-variety based measure of \textcite{cavallo2023product}, who find that welfare relative to the United States is lower than real consumption for many countries when deriving the cost of living from the \textcite{melitz2003impact} model (which like GEKS also assumes homotheticity).

\begin{figure}[h!]
    \centering
    \includegraphics[width=0.75\textwidth]{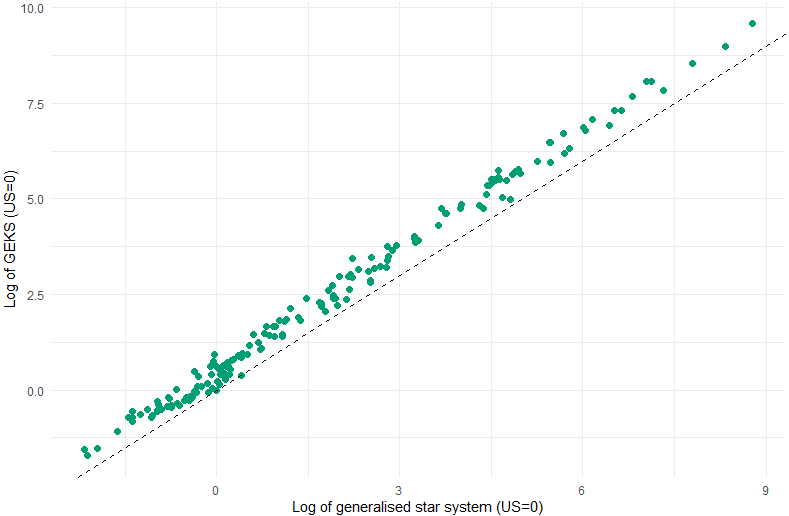}
    \caption{GSS and GEKS parities relative to the United States (log scale). Summary statistics for the percentage point differences between the parities are as follows. Mean -33.65; SD 13.93; Min -84.25 (Nepal, Niger); Max 121.76 (Montenegro, Argentina).} 
    \label{fig:geks_vs_gss_scatter}
\end{figure}
\begin{figure}[h!]
    \centering
    \includegraphics[width=0.75\textwidth]{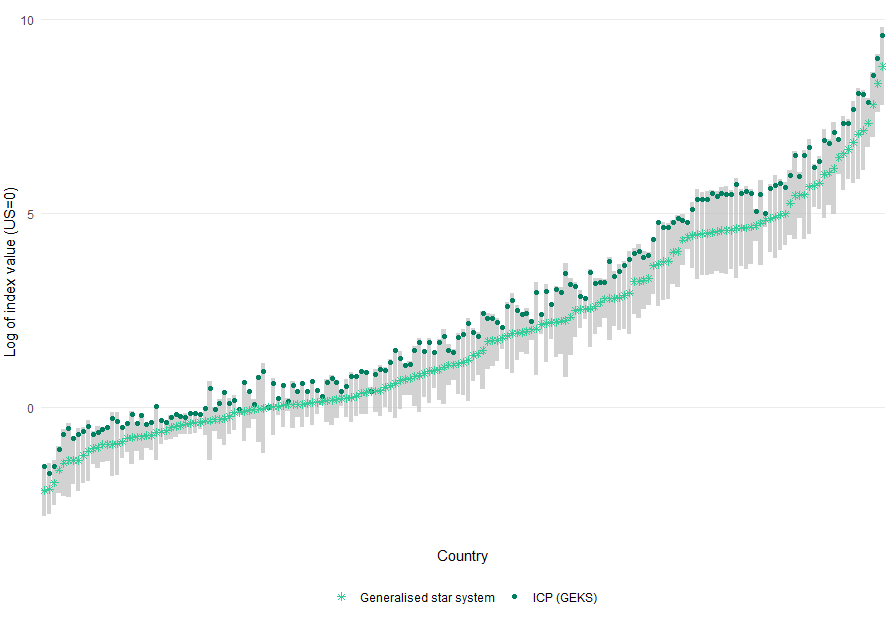}
    \caption{GSS and GEKS parities relative to the United States with the concomitant multilateral cost-of-living bounds (log scale). For nearly all countries in the data the GEKS parities lie very close to the upper non-parametric cost-of-living bound.} 
    \label{fig:geks_vs_gss_bounds_position}
\end{figure}

\paragraph{World output.} I now calculate world output by using the GSS parities to value expenditure for each country in 2017 using  observed consumption (i.e. the basic heading GSS PPPs, which are the geometric mean of the revealed preference cost-of-living bounds, are simply multiplied by the associated volumes). Given the distinctions between the two parities, the implied size of world output is higher when using the reference consumer's cost of living to value it compared to the ICP GEKS values. However, the output gap between the two methods is even larger than that suggested by the mean PPP difference: world output in US dollars under the generalised star system was 45.63\% larger in 2017 when compared to the ICP GEKS-based valuation. To put this figure in context, the magnitude of this difference is roughly comparable to the 49.97\% increase in the size of the world economy that year when calculated using PPPs versus market exchange rates (\cite{worldbank_ICP2017}).

\paragraph{Cross-country inequality.} With output based on GSS parities in hand, I now turn to analysing inter-country income inequality using these expenditures. Figure \ref{fig:geks_vs_gss_lorenz} compares the distribution of output per capita implied by the ICP and generalised-star PPPs in a pair of Lorenz curves showing the cumulative percentages of population against cumulative expenditure (beginning with the poorest economy), with perfect income inequality denoted by the 45-degree line. The world is markedly more equal when output is valued using the GSS parities. To interpret this difference in practical terms, the inter-country Gini coefficient, which reflects the area between the Lorenz curve and the equality line, falls from 0.442 for the ICP to 0.346 for the generalised-star and is a reduction comparable in magnitude to the pre- and post-tax change for OECD countries, and is roughly three times larger than the reduction in inequality between 2011 and 2017 measured by the ICP for the per-capita Actual Individual Consumption aggregate (\cite{worldbank_ICP2017}).

\begin{figure}[h!]
    \centering
    \includegraphics[width=0.75\textwidth]{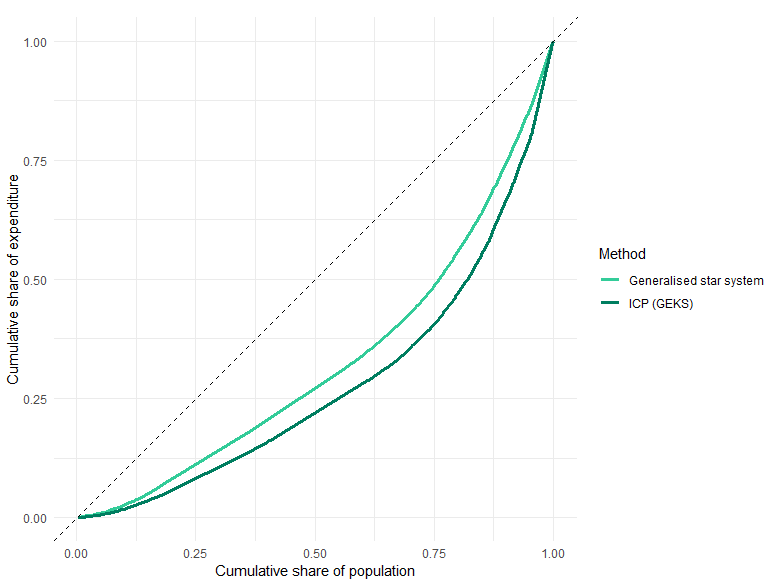}
    \caption{Lorenz curves for per-capita output using by the generalised star system PPPs (light green) and the GEKS parities underlying the standard ICP calculation (dark green). \textbf{Note:} The inter-country Gini coefficients for the respective parities are 0.442 and 0.346.} 
    \label{fig:geks_vs_gss_lorenz}
\end{figure}

Given the finding of the previous section that the GEKS index used by the ICP has a valid empirical welfare interpretation, a natural question about the difference between GEKS and the generalised-star statistics for output and inequality is `why should one prefer the GSS index if both itself and GEKS are welfare-consistent?' This question is important because the bounds themselves silent about the `best' index that satisfies them, while at the same time the differences between valid indices can materially change the value of global aggregates. In this light, perhaps the most compelling reason for choosing the GSS is because we do not know where in the bounds the true cost of living is located. Since it is based on utility, the cost of living is an unobservable object. The bounds from which the GSS are derived make minimal assumptions about the form of the reference consumer's tastes (aside from being rational and consistent with the data). Additionally, by choosing the geometric average, the resulting generalised-star system has an advantage in the presence of uncertainty about where in the permitted interval any true index lies since when the variance in the bounds' values is small relative to their scale, the GSS value will be close to the arithmetic mean (by the result of \cite{cartwright1978refinement}) which minimises loss under a naive ignorance prior. Further information about the form of tastes, for example through additional data, would of course help determine whether the cost of living lies close to a given index, but such a question is not within this paper's scope.

\section{Conclusion} \label{sec:discussion}

The welfare basis of international price and income comparisons is a critical issue for economic policy issues ranging from disbursing global aid flows to measuring living standards throughout the world. From the viewpoint of welfare however, many existing methods for conducting such comparisons can be `answers without questions' in the sense that the validity of the assumptions required for multilateral indices to have a cost-of-living interpretation can at times be unclear with respect to either or both the theory and data. The first main message of this paper is a comforting one in this regard: with the exception of market exchange rates, this problem is generally not a cause for concern because most popular comparison methods can be treated as if they possess a cost-of-living interpretation whether they intend to or not. The second relates to the major applications of the comparisons themselves: when measured using an index based on the midpoint of the non-homothetic cost of living bounds that is consistent with the data and incorporates income effects, the world was both larger and fairer in 2017 than was suggested by conventional estimates.

Despite attempting to clarify the welfare basis of international comparisons for existing methods, the methods I develop for doing so are not without their own limitations. Beyond the conceptual and methodological shortcomings already discussed, one major potential limitation of my two applications is that they are ultimately only as useful as the quality of the underlying dataset. Despite the seminal contribution of the ICP data to enabling cross-country price and income comparisons to exist at all, these data are known to possess a number of limitations with respect to their collection and measurement. Perhaps the most detailed decomposition of these measurement errors (spanning imputation, sampling, quality, and variety biases) is in \textcite{argente2023measuring} who leverage matched barcode-level data for the United States and Mexico to conduct their study. Using data free of such biases would improve the results, help further clarify the prevalence of the welfare inconsistencies identified in this paper, and improve the accuracy of the cost of living values from the generalised star system. 

There are several potentially fruitful extensions to the ideas in this paper. One is to broaden the welfare comparisons beyond the five indices considered to the multitude of others that exist in the literature: doing so would shed additional light on the impact of other theoretical properties on indices' welfare validity. Relatedly, while the focus of my appraisal application was primarily on widely-used non-parametric price index formulae, as demonstrated by the inclusion of market exchange rates, many other objects could be studied using the logic of this paper. For example, the bounds in Proposition \ref{prop:nonparam_col_bounds} and associated generalised star values could be compared against parametric welfare measures such as those in \textcite{jones2016beyond} or \textcite{cavallo2023product}. Such comparisons may begin to depart from the main research question of this study as the data underlying such indices may not be like-for-like, but this is not a fatal flaw. For example, the revealed preference test in Proposition \ref{prop:ref_consumer_characterisation} could be extended to account for data other than price and quantities through concomitant results about general budget sets in the revealed preference literature (e.g. \cite{matzkin1991axioms}). Beyond welfare measures, while the research question in this paper concerns comparisons of welfare over space, the methods in this paper could also be applied to problems relating to measuring the cost of living over time such as taste shocks and the changing structure of consumption and production. Future work could also investigate the determinants of the heterogeneity between countries documented by the empirical results of the paper such as the correlates and reasons for why certain countries exhibit the error rates, bound improvements, and movement in welfare ranks that they are found to have.

\appendix
\section{Proofs}
\label{sec:proofs}

\subsection{Proof of Proposition \ref{prop:ref_consumer_characterisation}}
\label{proof:ref_consumer_characterisation}

\textbf{($\Rightarrow$)} Suppose there exists a reference consumer for $\mathcal{D}$ with the tastes $u_{RC}(\mathbf{q})$ described in 1. Consider any cycle in $G$ for these data where $w(\upsilon_\ell, \upsilon_{\ell+1}) \leq 1$ for $\ell = 1, 2, \dots, k$. Along these cycles we have $\mathbf{p}_{\ell}\cdot\mathbf{q}_{\ell+1}\leq\mathbf{p}_{\ell}\cdot\mathbf{q}_{\ell}$ where by Definition \ref{def:ref_consumer} the utility function $u_{RC}(\mathbf{q})$ assigns $u_{RC}(\mathbf{q_{\ell+1}}) \leq u_{RC}(\mathbf{q_{\ell}})$.

I now claim that $\mathbf{p}_{\ell}\cdot\mathbf{q}_{\ell+1}=\mathbf{p}_{\ell}\cdot\mathbf{q}_{\ell}$ for all $\ell$. Suppose this was not true so that for some pair along the cycle we observe $\mathbf{p}_{m}\cdot\mathbf{q}_{m+1}<\mathbf{p}_{m}\cdot\mathbf{q}_{m}$. Non-satiation of $u_{RC}(\mathbf{q})$ with the fact that $u_{RC}(\mathbf{q})$ rationalises $\mathcal{D}$ means that we must have the strict $u_{RC}(\mathbf{q}_{m+1}) < u_{RC}(\mathbf{q}_{m})$, for if we had $u_{RC}(\mathbf{q}_{m+1}) = u_{RC}(\mathbf{q}_{m})$ we would be able to locate a nearby affordable bundle $\mathbf{p}_{m}\cdot\mathbf{q}_{m+\epsilon}<\mathbf{p}_{m}\cdot\mathbf{q}_{m}$ with $u_{RC}(\mathbf{q}_{m+\epsilon}) > u_{RC}(\mathbf{q}_{m+1})=u_{RC}(\mathbf{q}_{m})$ contradicting the fact that $u_{RC}(\mathbf{q})$ rationalises the observed choice $\mathbf{q}_{m}$. 

In turn, having $u_{RC}(\mathbf{q}_{m+1}) < u_{RC}(\mathbf{q}_{m})$ at some point within the considered cycle is problematic because closing the cycle yields the impossibility $u_{RC}(\mathbf{q}_{m}) > u_{RC}(\mathbf{q}_{m+1}) \geq u_{RC}(\mathbf{q}_{m+2}) \geq \cdots \geq u_{RC}(\mathbf{q}_{m})$. We conclude that the wedge weights $w(\upsilon_\ell, \upsilon_{\ell+1}) = 1$ for all $\ell$ which establishes CEWEC. Notice that the other properties of $u_{RC}(\mathbf{q})$ are required for this argument.

\textbf{($\Leftarrow$)} To begin, observe that the CEWEC is a graph-theoretic re-statement of the `cyclical consistency' condition of \textcite{afriat_construction_1967}. Cyclical consistency is defined as $\mathbf{p}_i \cdot \mathbf{q}_{i+1} \geq \mathbf{p}_i \cdot \mathbf{q}_i \; \forall i = 1, \dots, k \; (\text{mod } k) \Rightarrow \mathbf{p}_i \cdot \mathbf{q}_{i+1} = \mathbf{p}_i \cdot \mathbf{q}_i \; \forall i$ which is obviously equivalent to the CEWEC. 

As a result, the remainder of the argument is the standard, substantive step of Afriat's Theorem for which a large number of proofs exist; for this reason none are repeated here. For example, in addition to the original constructive proof of \textcite{afriat_construction_1967} this step can also be derived through linear programming (\cite{diewert1973afriat}) and even via the Minimax Theorem of two-player zero-sum games (\cite{geanakoplos2013afriat}). \hfill $\blacksquare$	

\subsection{Proof of Proposition \ref{prop:nonparam_col_bounds}}
\label{proof:nonparam_col_bounds}

I prove Proposition \ref{prop:nonparam_col_bounds} in two steps: I first (1) bound the expenditure function, and then (2) use this result to obtain bounds for the cost of living.

\textbf{Step 1: obtain revealed preference bounds on the expenditure function}. Let us first consider the problem of bounding the reference consumer's expenditure function. These bounds are the quantities in \eqref{eq:M_minus} and \eqref{eq:M_plus} within Proposition \ref{prop:nonparam_col_bounds} which I repeat here for convenience:
    \begin{align}
        M^{-}(\mathbf{p},\upsilon) &= \inf_{\mathbf{q}} \mathbf{p}.\mathbf{q} \text{ such that } \mathbf{p}_{i}\mathbf{q}\geq \mathbf{p}_{i}\mathbf{q}_{i} \text{ for all } i \in VRW(\upsilon) \label{eq:M_minus_v2} \\
        M^{+}(\mathbf{p},\upsilon) &= \min_{\mathbf{q}} \mathbf{p}.\mathbf{q}_{i} \text{ such that } i \in VRP(\upsilon) \label{eq:M_plus_v2}
    \end{align}
I first prove a lemma showing that $M^{-}(\mathbf{p},\upsilon)$ and $M^{+}(\mathbf{p},\upsilon)$ do indeed bound the reference consumer's expenditure function. This result is well-known within the revealed preference literature. My approach follows \textcite{varian_nonparametric_1982} and \textcite{knoblauch_tight_1992} who were the first to conjecture and prove the result. An equivalent statement using a procedure involving Engel curves is at \textcite{blundell2003nonparametric}.
\begin{lemma} \label{lemma:exp_nonparam_bounds}
\textbf{(\cite{varian_nonparametric_1982}, Fact 11(i); Fact 5)} For data $\mathcal{D}$ satisfying Proposition \eqref{prop:ref_consumer_characterisation} it holds that
    \begin{align}
        M^{-}(\mathbf{p},\upsilon) \leq e(\mathbf{p},u_{RC}(\mathbf{q}_\upsilon))) \leq M^{+}(\mathbf{p},\upsilon)
        \label{eq:exp_bounds}
    \end{align}
\end{lemma}
\begin{proof}
    To begin, I first enlarge two of the sets in Definition \ref{def:rp_sets} to allow for a potentially unobserved country with data $(\mathbf{p},\mathbf{q})$. To this end, define the extended graph $G_{G\cup \mathbf{q}}$ as $G$ with a (potentially) new vertex $\upsilon_{\mathbf{q}}$, and (2) the $N$ (potentially) new weighted directed edges terminating at $\upsilon_{\mathbf{q}}$. Now consider the expanded strict analogues of $VRP(\upsilon)$ and $VRW(\upsilon)$,
    \begin{align}
    \operatorname{QRP}(v) &= \{\,\mathbf{q} : \upsilon P \upsilon_{\mathbf{q}} \,\}, \label{eq:QRP}\\
    \operatorname{QRW}(v) &= \{\,\mathbf{q} : \upsilon_{\mathbf{q}} P \upsilon \,\}. \label{eq:QRW}
    \end{align}
We can now re-define $M^{-}(\mathbf{p},\upsilon)$, $M^{+}(\mathbf{p},\upsilon)$, and $e(\mathbf{p},u_{RC}(\mathbf{q}_\upsilon)))$ in the convenient forms
\begin{align}
    M^{-}(\mathbf{p},\upsilon) &= \inf \mathbf{p}.\mathbf{q} \text{ such that } \mathbf{q}\in QRW^{c}(\upsilon) \\
    M^{+}(\mathbf{p},\upsilon) &= \inf \mathbf{p}.\mathbf{q} \text{ such that } \mathbf{q}\in QRP(\upsilon) \\
    e(\mathbf{p},u_{RC}(\mathbf{q}_{\upsilon})) &= \inf \mathbf{p}.\mathbf{q} \text{ such that } \mathbf{q} \in E(\upsilon)
\end{align}
where $E(\upsilon)\equiv \{\mathbf{q}|u_{RC}(\mathbf{q})>u_{RC}(\mathbf{q}_{\upsilon})\}$. Next, I appeal to the reference consumer's preferences to link elements of $QRW(\upsilon)$ and $QRP(\upsilon)$ with their associated utility levels. To make this link, observe that for any qualifying walk where $\upsilon P \upsilon_{i}$, it must be the case that $u_{RC}(\mathbf{q}_{\upsilon})>u_{RC}(\mathbf{q}_{i})$ as a result of $\mathcal{D}$ being rationalisable by $u_{RC}(\mathbf{q}_{\upsilon})$ via Proposition \ref{prop:ref_consumer_characterisation}. By this reasoning, any $\mathbf{q}\in QRW(\upsilon)$ has the property that $u_{RC}(\mathbf{q})< u_{RC}(\mathbf{q}_{\upsilon})$. Similarly, for any $\mathbf{q}\in QRP(\upsilon)$ we have $u_{RC}(\mathbf{q})>u_{RC}(\mathbf{q}_{\upsilon})$.

The remainder of the argument follows an identical logic to \textcite{varian_nonparametric_1982}. I claim that
\begin{align*}
    QRP(\upsilon) \subseteq E(\upsilon) \subseteq QRW^{c}(\upsilon).
\end{align*}
Examining $QRP(\upsilon) \subseteq E(\upsilon)$, since any $\mathbf{q}\in QRP(\upsilon)$ has $u_{RC}(\mathbf{q})>u_{RC}(\mathbf{q}_{\upsilon})$, so $\mathbf{q}\in E(\upsilon)$ as well. In contrast, any $\mathbf{q} \in E(\upsilon)$ has $u_{RC}(\mathbf{q})>u_{RC}(\mathbf{q}_{\upsilon})$, which is not compatible with being in $QRW(\upsilon)$. Thus $\mathbf{q}\in QRW^{c}(\upsilon)$ and $E(\upsilon) \subseteq QRW^{c}(\upsilon)$. To conclude, as the infimum of a set can only be smaller or equal to that of a subset of that set we have
\begin{align*}
    \inf \{\mathbf{p}.\mathbf{q}:\mathbf{q}\in QRP(\upsilon)\} 
    \geq \inf \{\mathbf{p}.\mathbf{q}:\mathbf{q}\in E(\upsilon)\} \geq \inf \{\mathbf{p}.\mathbf{q}:\mathbf{q}\in QRW^{c}(\upsilon)\}
\end{align*}
which yields the desired result.
\end{proof}
An important fact about \eqref{eq:M_minus_v2} and \eqref{eq:M_plus_v2} is that they are the tightest possible bounds on $e(\mathbf{p},u_{RC}(\mathbf{q}_\upsilon)))$. This result was proven by \textcite{knoblauch_tight_1992} so is omitted here.

\textbf{Step 2: obtain multilateral revealed preference bounds on the cost of living}. I am now in a position to bound the cost of living. To this end, consider \eqref{eq:bound_improvements_laspeyres}. I first argue that $\frac{M^{+}(\mathbf{p}_{j},\upsilon_{i})}{M^{+}(\mathbf{p}_{i},\upsilon_{i})} \leq \frac{\mathbf{p}_{j}.\mathbf{q}_{i}}{\mathbf{p}_{i}.\mathbf{q}_{i}}$. Since $M^{+}(\mathbf{p}_{i},\upsilon_{i})=m_{i}$, it suffices to show that $M^{+}(\mathbf{p}_{j},\upsilon_{i}) \leq \mathbf{p}_{j}.\mathbf{q}_{i}$. Suppose there exists a vertex $k\in VRP(\upsilon_{i})$ such that $M^{+}(\mathbf{p}_{j},\upsilon_{i})=\mathbf{p}_{j}.\mathbf{q}_{k}>\mathbf{p}_{j}.\mathbf{q}_{i}$. Now, note that $i \in VRP(\upsilon_{i})$ since the edge weight $w(\upsilon_{i},\upsilon_{i})=1$. Then $\mathbf{q}_{k}$ cannot be the bundle defining $M^{+}(\mathbf{p}_{j},\upsilon_{i})$ since $\mathbf{q}_{i}$ is also feasible and produces a strictly smaller expenditure than $\mathbf{q}_{k}$ at $\mathbf{p}_{j}$.

Next I show that $\min\left\{{\frac{p_{j}^{k}}{p_{i}^{k}}}\right\} \leq \frac{M^{-}(\mathbf{p}_{j},\upsilon_{i})}{M^{-}(\mathbf{p}_{i},\upsilon_{i})}$. First, rewrite the right hand side of the inequality as
    \begin{align} \label{eq:prop1_weight_step}
        \frac{M^{-}(\mathbf{p}_{j},\upsilon_{i})}{M^{-}(\mathbf{p}_{i},\upsilon_{i})}=\frac{\mathbf{p}_{j}.\mathbf{q}_{M^{-}}}{m_{i}}
    \end{align}
Multiplying by $(p_{i}/p_{i})$ produces a weighted sum of price relatives
    \begin{align}
        \frac{\mathbf{p}_{j}.\mathbf{q}_{M^{-}}}{m_{i}}=
        \sum^{K}_{k=1}\left(\frac{p_{j}^{k}}{p_{i}^{k}}\right)\left(\frac{q_{M^{-}}^{k}p_{i}^{k}}{m_{i}}\right)
    \end{align}
Since $i \in VRP(\upsilon_{i})$ it must be true that $\mathbf{p}_{i}\mathbf{q}_{M^{-}}\geq \mathbf{p}_{i}\mathbf{q}_{i}$. It follows that the sum of weights in (15) $\sum^{K}_{k} \frac{q_{M^{-}}^{k}p_{i}^{k}}{m_{i}}\geq \sum^{K}_{k} \frac{q_{i}^{k}p_{i}^{k}=1}{m_{i}}=1$, with the latter set of weights clearly being greater than or equal to $\min\{(p_{j}^{k}/p_{i}^{k})\}$.

It remains to show that the inner inequalities $\frac{M^{-}(\mathbf{p}_{j},\upsilon_{i})}{M^{-}(\mathbf{p}_{i},\upsilon_{i})} \leq \frac{e(\mathbf{p}_{j},u_{RC}(\mathbf{q}_{i}))}{e(\mathbf{p}_{i},u_{RC}(\mathbf{q}_{i}))} \leq \frac{M^{+}(\mathbf{p}_{j},\upsilon_{i})}{M^{+}(\mathbf{p}_{i},\upsilon_{i})}$ hold. This follows immediately from Lemma \ref{lemma:exp_nonparam_bounds} as each term's denominator is equal to $m_{i}$. 
    
The proof of \eqref{eq:bound_improvements_paasche} — the Paasche-style counterpart to \eqref{eq:bound_improvements_laspeyres} — follows a similar logic to the Laspeyres-style case of \eqref{eq:bound_improvements_laspeyres} so is omitted. \hfill $\blacksquare$	

\subsection{Proof of Fact \ref{fact:bounds_tighten_w_n}}
\label{proof:bounds_tighten_w_n}

Suppose \eqref{eq:more_obs1} or \eqref{eq:more_obs2} were not true. This implies the removal of a country walk from either $VRW(\upsilon)$ or $VRP(\upsilon)$ from $G$. This is clearly impossible since $G'$ contains $G$ as a subgraph so possesses at least all of $G$'s walks. \hfill $\blacksquare$	

\printbibliography
\newpage
\section{Supplemental Appendix}

\subsection{Further notes about the data and their construction}
\label{subsec:further_data_notes}

\paragraph{1. Converting PPPs and nominal expenditures into prices and quantities.} It is straightforward to transform these raw data into the prices and quantities $\mathbf{p}_{n}$ and $\mathbf{q}_{n}$ described in the setting in Section \ref{sec:notation}. As is standard in the international comparisons literature (see for example \cite{diewert_axiomatic_1999} and \cite{diewert_methods_2013}), we simply adopt the unit of a base currency — in this case the US dollar — as a common numeraire across goods. This imparts the PPP $\pi^{k}_{ib}$ the interpretation as a price \textit{level} for good $k$ in units of country $b$'s currency which in turn yields the implicit quantities or volumes $q^{k}_{i}\equiv e_{i}^{k}/\pi^{k}_{ib}=p_{b}^{k}r_{i}^{k}$ (with $r_{i}^{k}$ denoting quantities in the units of the basic heading $k$) again expressed in the units of the numeraire currency. Thus, the Laspeyres quantity index (for example, as in \eqref{eq:paasche_edge_weight}) between any desired country pair can be expressed through PPPs and nominal expenditures alone since $\mathbf{p}_{m}.\mathbf{q}_{n}/\mathbf{p}_{m}.\mathbf{q}_{m}=\sum^{K}_{k}(\pi_{mb}^{k}/\pi_{nb}^{k})e_{n}^{k}/\sum^{K}_{k}e_{m}^{k}$.

\paragraph{2. Rationale for not linking further ICP waves.} My main application uses prices and per-capita demands from the ICP’s 2017 reference year for a subset of basic headings corresponding to consumption. Although the ICP has collected data from as early as 1970, and although Proposition \ref{prop:ref_consumer_characterisation} can easily be adapted to panels as well as to cross-sections, it is generally not appropriate to pool data across ICP waves due to the differences in the production of the statistics over time. For example, the 2017 reference year utilised a different national accounts framework compared to previous waves and ICP waves prior to 2005 adopted a different taxonomy of basic headings and method for linking regional results relative to previous waves. These issues are documented in \textcite{worldbank_ICP2017} and reflect the fact that the ICP is an ongoing initiative subject to continued methodological improvements. 

\paragraph{3. Excluded basic headings.} With respect to excluded basic headings, the 155 categories in the data include sub-aggregates related to government spending, investment, and the trade balance, in addition to consumption. Since revealed preference theory only applies to the last of these categories, I exclude the others from the sample. The resulting set of consumption sub-aggregates is a standard statistics aggregate known as Household and NPISHs' Final Consumption Expenditure. I also exclude basic headings with negative expenditures and countries with missing prices at the basic heading level: doing so is necessary for a valid test of Proposition \ref{prop:ref_consumer_characterisation} and is standard practice in studies using revealed preference techniques (see for example \cite{dowrick_international_1994}, who use data from the ICP's 1980 reference year). 

\subsection{Multilateral indices selected in the appraisal application} \label{subsec:index_overview}

\begin{table}[h!]
\centering
\renewcommand{\arraystretch}{1.3} 
\begin{tabularx}{\textwidth}{>{\raggedright\arraybackslash}p{5.5cm} p{9cm}}
\toprule
\toprule
\textbf{Index number method} & \textbf{Example application} \\
\midrule
\makecell[l]{\textbf{Geary-Khamis (GK)} \\[4pt]
$P_{n}^{GK} = \dfrac{\sum_{k}\pi_{k}q_{nk}}{\sum_{k}p_{nk}q_{nk}}$ \\[4pt]
$\pi_{k} = \dfrac{\sum_{n}P_{n}^{GK}m_{nk}}{\sum_{n}q_{nk}}$
}
& Used to construct real GDP measures within the Penn World Table (\cite{feenstra2015next}) and in the ICP prior to 1990. \\
\midrule
\makecell[l]{\textbf{GEKS} \\[4pt]
$P_{ij}^{GEKS} = \left(\prod_{k=1}^{N}\frac{P_{jk}^{Fisher}} {P_{ki}^{Fisher}}\right)^{\frac{1}{N}}$ \\[4pt] where \\[4pt]
$P_{ij}^{Fisher}=\left(
            \frac{\mathbf{p}_{i}\cdot\mathbf{q}_{j}}{\mathbf{p}_{j}\cdot\mathbf{q}_{j}}
            \cdot
            \frac{\mathbf{p}_{j}\cdot\mathbf{q}_{i}}{\mathbf{p}_{i}\cdot\mathbf{q}_{j}}\right)^{\frac{1}{2}}$
}
& Used to construct modern PPPs and real GDP measures within the ICP after 1990 (e.g. see \cite{diewert_methods_2013}). \\
\midrule
\makecell[l]{\textbf{CCD} \\[4pt]
$P_{ij}^{CCD} = \left(\prod_{k=1}^{N}\frac{P_{ik}^{T\mathaccent"7F o rnqvist}}{P_{kj}^{T\mathaccent"7F o rnqvist}}\right)^{\frac{1}{N}}$ \\[4pt] where \\[4pt]
$P_{ij}^{T\mathaccent"7F o rnqvist} = 
\prod_{k=1}^{K} \left(\frac{p_{jk}}{p_{ik}}\right)^{\frac{1}{2}(s_{ik}+s_{jk})}$ \\[4pt]
$s_{ik}= \dfrac{p_{ik}q_{ik}}{\sum^{K}_{m=1}p_{im}q_{im}}$
}
& Used to compare cross-country output (e.g. \cite{rao1995generalized}) and productivity (e.g. \cite{fare1994productivity}). \\
\bottomrule
\bottomrule
\end{tabularx}
\caption{A summary of the appraised multilateral indices and their uses.}
\label{tab:index_overview}
\end{table}

\subsection{Correcting indices for preference misspecification}
\label{subsec:taste_corrected_indices}

It may sometimes be desirable to use an index for applications even if it is not fully welfare-consistent. To support such scenarios, here I show how to construct a taste‑corrected version of any multilateral index by correcting for errors from preference misspecification.

My solution is extremely simple and is grounded on the fact that whenever a country‑pair index violates one of its theoretical bounds
(\(\mathcal{M}^{-}_{ij},\mathcal{M}^{+}_{ij}\)),
the unknown distance to the true, bound‑consistent value must lie within the interval
\([\mathcal{M}^{-}_{ij},\mathcal{M}^{+}_{ij}]\).
Under symmetric measurement error or a uniform ignorance prior, the expected value of the true index is the arithmetic midpoint
\(\tilde{\mathcal{M}}_{ij}\equiv(\mathcal{M}^{-}_{ij}+\mathcal{M}^{+}_{ij})/2\).
I therefore replace any offending observation with \(\tilde{\mathcal{M}}_{ij}\),
which restores welfare consistency and minimises the mean‑squared and mean‑absolute errors relative to the unobserved true bound under a uniform density, i.e.
\begin{align}
\tilde{\mathcal{M}}
    \;=\;
    \arg\min_{\hat m}
    \int_{\mathcal{M}^{-}}^{\mathcal{M}^{+}}
         \frac{|m-\hat m|^{\,p}}{\mathcal{M}^{+}-\mathcal{M}^{-}}\,
         dm
    \qquad p\in\{1,2\}
    \label{eq:error_integral}
\end{align}
where $p=1$ and $2$ respectively gives the mean‑absolute and mean-squared error criterion.

Table \ref{tab:taste_corrected_geks} demonstrates this idea with the taste-corrected observations for the GEKS index with the ICP convention of using the United States as the base country. Reflecting the findings of the earlier appraisal application, only a handful of countries require correction to be consistent with an admissable cost-of-living interpretation and the magnitude of these corrections is small in magnitude.
\begin{table}[ht]
\centering
\begin{tabular}{lrrrr}
\toprule
\toprule
Country & GEKS & $\mathcal{M}^{-}$ & $\mathcal{M}^{+}$ & Corrected GEKS ($\tilde{\mathcal{M}}$)\\ 
\hline
  Austria (AUT) & 0.83 & 0.56 & 0.82 & 0.69 \\ 
  Switzerland (CHE) & 1.33 & 1.04 & 1.30 & 1.17 \\ 
  Curaçao (CUW) & 1.49 & 0.57 & 1.46 & 1.02 \\ 
  Czech Republic (CZE) & 13.90 & 5.71 & 13.78 & 9.75 \\ 
  France (FRA) & 0.82 & 0.46 & 0.81 & 0.64 \\ 
  Greece (GRC) & 0.66 & 0.31 & 0.65 & 0.48 \\ 
  Japan (JPN) & 116.99 & 58.54 & 106.92 & 82.66 \\ 
  South Korea (KOR) & 1003.21 & 403.90 & 947.83 & 675.87 \\ 
  Taiwan (TWN) & 16.56 & 9.59 & 16.35 & 12.97 \\ 
\bottomrule
\bottomrule
\end{tabular}
\caption{The taste-corrected GEKS index for PPPs relative to the United States for observations violating the reference consumer's cost-of-living bounds.}
\label{tab:taste_corrected_geks}
\end{table}

\subsection{Analogous results under homothetic tastes} \label{subsec:harp}

The price index defined in Section \ref{subsec:fat_star_og} is transitive when an indifference level $u(\mathbf{q})$ is fixed, but many such levels could be chosen for the purposes of the desired comparison with different index results as a consequence of the non-homothetic tastes underlying the index. For many applications it is desirable to construct preference-based price indices that do not exhibit such dependence on an indifference level (or for quantity indices, dependence on an associated price vector). As has long been known (\cite{malmquist1953index}; \cite{samuelson1974invariant}; \cite{Pollak1971a}) homothetic tastes can produce such an index because the expenditure function can be decomposed into the form $f(\mathbf{q})g(u(\mathbf{q}))$. It is straightforward to extend the ideas in Section \ref{sec:intlrefcons} to the special case of homothetic preferences and given this, I now the homothetic analogues to the main concepts and results in this paper for readers interested in indices of this form.

\textbf{Multilateral indices based on homothetic revealed preference}: The homothetic analogue to the non-parametric bounds in Proposition \ref{prop:nonparam_col_bounds} are not new within the revealed preference literature. The homothetic analogue to GARP — the Homothetic Axiom of Revealed Preference (HARP) — was first discovered by \textcite{diewert1973afriat}. This restriction and the associated cost-of-living bounds (\cite{afriat_constructability_1981}) are the basis of the multilateral Afriat Ideal Index introduced by \textcite{dowrick_true_1997}. HARP defines a reference consumer analogous to that in Proposition \ref{prop:ref_consumer_characterisation} which is easily described in graph theoretic terms:
\begin{theorem}
    \textbf{(\cite{diewert1973afriat}; \cite{afriat_constructability_1981}; \cite{varian_non-parametric_1983})} There exists a nonsatiated homothetic utility function that rationalizes the data if and only if for all distinct choices of the indexes $(i,j,...,m)$ we have
    \begin{align}
    w_{ij}(w_{jk})...(w_{mi}) \geq 1 
    \label{eq:harp_weight_product}
    \end{align}
    for all closed walks $\upsilon_{i} (\upsilon_{i},\upsilon_{j})...(\upsilon_{m},\upsilon_{i}) \upsilon_{i}$ in $G$.
    \label{theorem:harp}
\end{theorem}
Turning to cost-of-living bounds based on these tastes, the existing literature (most notably \textcite{afriat_constructability_1981}, \textcite{manser_analysis_1988}, and \textcite{dowrick_true_1997}) derives bounds based on a minimum path matrix related to \eqref{eq:harp_weight_product} from which a variety of multilateral indices can be constructed. As these results are all established I refer the reader to these sources for further exposition of them.

Multilateral homothetic indices based on revealed preference are thus not new. However, as it currently stands, these indices' values are not correctly defined in the existing literature for countries outside of this set. Only one proposal exists in the literature for doing this: \textcite{dowrick_true_1997} suggest using the minimum of the bilateral Laspeyres valuations for any observation outside of the reference consumer set (p. 49). The problem with this idea is that while the Laspeyres valuation does indeed bound a cost-of-living function, the tastes underlying this function are not the same as the ones for the reference consumer underlying the countries that jointly admit a homothetic representation. I solve this problem by computing an index value based explicitly on the reference consumer's tastes using the logic of the generalised star system introduced in Section \ref{subsec:fat_star_og} which in turn shows how to correctly extend the Ideal Afriat Index to non-reference consumer countries. Specifically, I implement the generalised star system in Section \ref{subsec:fat_star_og} for homothetic tastes using the insight that the multiplicative condition in \eqref{eq:harp_weight_product} is equivalent to CEWEC holding for an augmented dataset, which I denote $\mathcal{D}_{H}$ comprising $\mathcal{D}$ plus a finite set of observations implied by the linear income expansion paths required by homotheticity. Algorithms to compute these implied observations were developed by \textcite{blundell2015sharp}. With $\mathcal{D}_{H}$ in hand, the homothetic cost-of-living bounds and feasible index values follow directly from Proposition \ref{prop:nonparam_col_bounds} and \eqref{eq:M_plus_outside_star} regardless of the whether a country's budget satisfies Theorem \ref{theorem:harp}.

\textbf{The homothetic reference consumer}: Unlike the results in the main application, my choice of homothetic reference consumer is not the largest possible set of countries that satisfy Theorem \ref{theorem:harp}. Since many countries do not admit a homothetic representation, such a group is not easy to find computationally. In lieu of this fact and to nonetheless identify a reference consumer taste of broad economic significance, I construct the reference consumer group using a procedure beginning with the greedy group-building algorithm based on price similarity described below. Algorithm \ref{alg:harp_rc_greedy_algo} starts with each of the 173 countries in the data, sorts the remaining countries by their cosine price similarity relative to the starting country, and then attempts to build a reference consumer group by iteratively adding countries to the group and testing their data against Theorem \ref{theorem:harp}. Countries that pass are added to the group for the next stage. The rationale for constructing the reference consumer's tastes based on a price similarity measure is twofold, and also has precedence in the empirical revealed preference literature (\cite{crawford_how_2013}). First, as was pioneered in the index number literature by \textcite{hill1999comparing} and noted in related work (for example, \textcite{hill2004constructing} and \textcite{diewert_methods_2013}), if the price structures of two countries are similar, then their cost-of-living bounds will be more accurate in the sense of the Paasche and Laspeyres indicies being closer to each other than would otherwise be the case. Second, price similarity also maximises the `power' of the underlying revealed preference tests by maximising the region of the budget plane that can lead to a rejection of Theorem \ref{theorem:harp}. After using Algorithm \ref{alg:harp_rc_greedy_algo} to generate $N=173$ candidate reference consumers, in line with the main application I then select the set that accounts for the highest share of world GDP and population among these candidates.

\singlespacing
\begin{algorithm}[H]
    \caption{Greedy Reference Consumer Candidate Construction}\label{alg:harp_rc_greedy_algo}
    \begin{algorithmic}[1]
        \Require Datasets $\mathcal{D}_{i}$ for $N$ countries, indexed by $i = 1,\dots, N$
        \For{$i = 1$ to $N$}
            \State Sort countries in order of descending cosine price similarity relative to $\mathbf{p}_i$
            \State Index the sorted order by $m = 1, \dots, M$
            \State Create candidate reference consumer group $RC_{i} \gets \varnothing$
            \For{$m = 1$ to $M$}
                \If{$\mathcal{D}_{m} \cup \left( \bigcup_{i \in RC_{i}} \mathcal{D}_{i} \right)$ satisfies Theorem~\ref{theorem:harp}}
                    \State $RC_{i} \gets RC_{i} \cup \{m\}$
                \Else
                    \State \textbf{continue} 
                \EndIf
            \EndFor
        \EndFor
        \State \Return $\{RC_{1}, \dots, RC_{N}\}$
    \end{algorithmic}
\end{algorithm}
\doublespacing

The seed country for the reference consumer set is Ireland. The countries included in this set are shown in Figure \ref{fig:harp_rc_map}. Countries outside of the reference consumer set tend to geographically across a north-west to south-east `sweep' centered on the African continent.
\begin{figure}[h!]
    \centering
    \includegraphics[width=\linewidth]{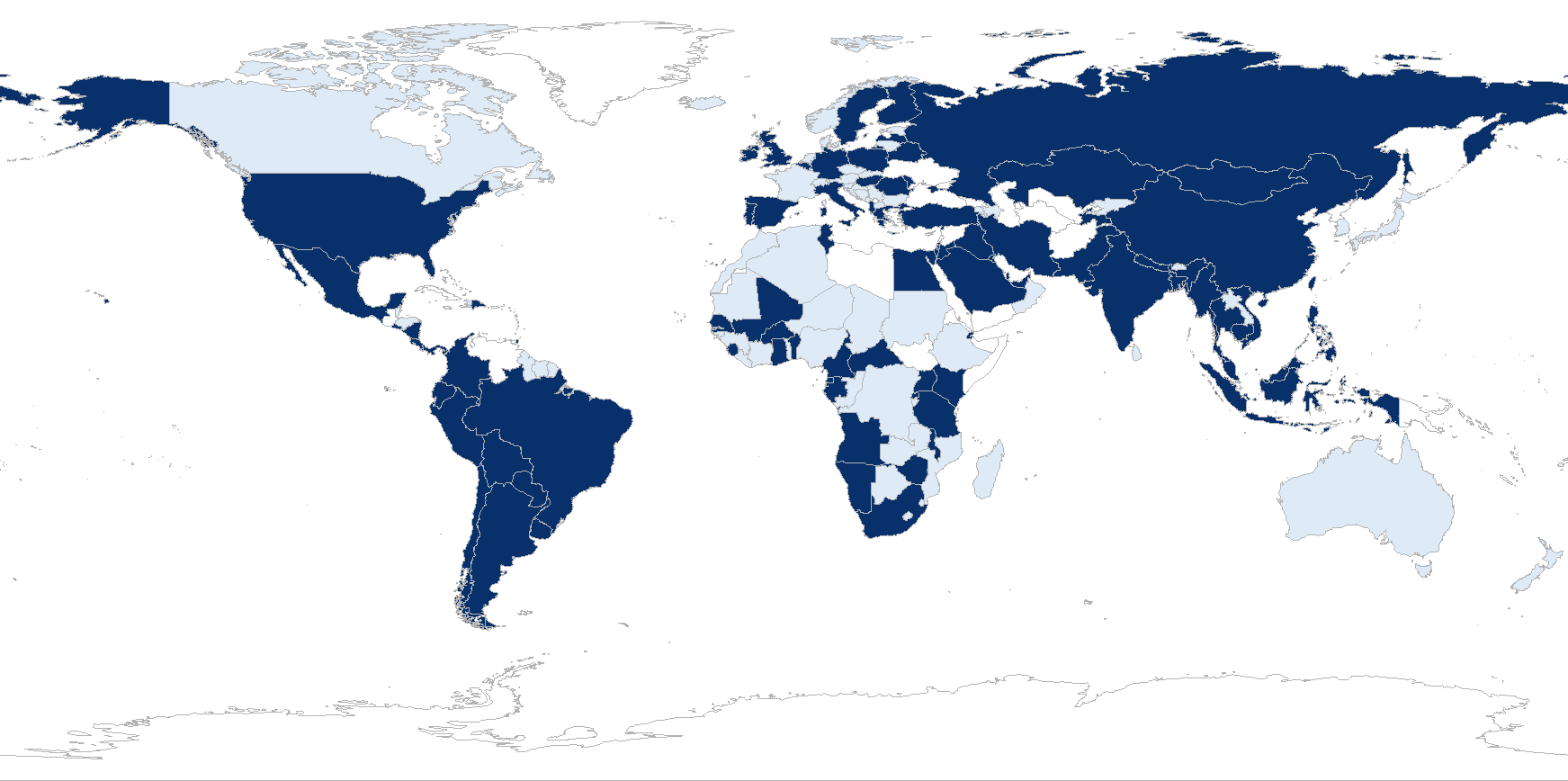}
    \caption{The 103 countries in dark blue are those spanned by the distinguished homothetic reference consumer in the ICP's 2017 reference year. These countries collectively account for 84\% of the world's 2017 aggregate output and population among countries in the sample.}
    \label{fig:harp_rc_map}
\end{figure}
Despite the restrictions on behaviour implied from homotheticity, the chosen reference consumer remains large in that it accounts for 84\% of the world's 2017 aggregate output and population among countries for which there are data. This finding is consistent with the analysis of \textcite{dowrick_true_1997} who find that their entire sample admits a homothetic index, albeit only for the 17 countries that comprised their final sample based on available data at the time.

\textbf{Tightened bilateral bounds under homotheticity}: Multilateral improvements to the upper and lower classical bilateral homothetic bounds (the Paasche and Laspeyres indices) for countries in the reference consumer set are shown at Figure \ref{fig:bound_improv_violin_HARP}. In contrast to the case of non-homothetic tastes in Proposition \ref{prop:nonparam_col_bounds} and Figure \ref{fig:bound_improv_violin_GARP}, the tightening of the bounds is more symmetrically distributed between the classical lower and upper bounds. Additionally, the improvements to the upper cost-of-living bounds (which is the Laspeyres price index in both the homothetic and non-homothetic cases) are larger in magnitude for the homothetic case -- despite there being 64\% fewer country observations in the reference consumer set -- and are no longer concentrated at zero.
\begin{figure}[h!]
    \centering
    \includegraphics[width=\linewidth]{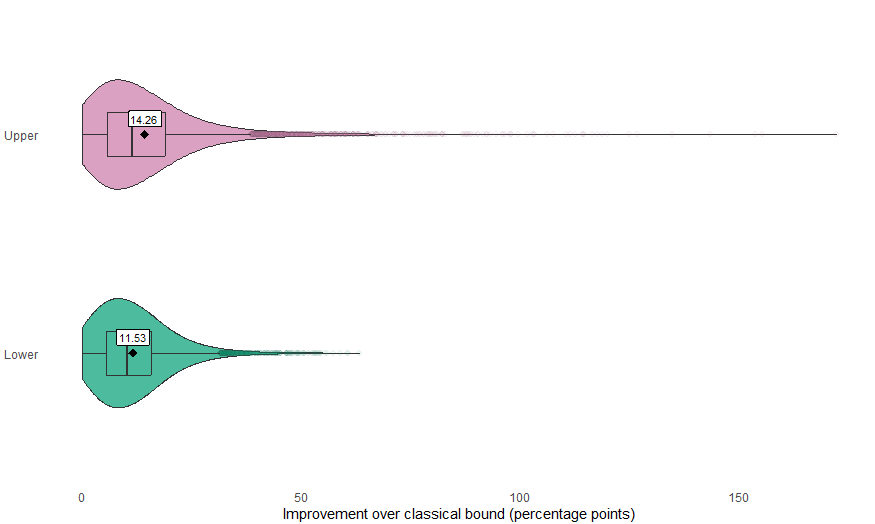}
    \caption{Improvements between the classical bilateral and non-parametric multilateral bounds under homothetic preferences for the reference consumer shown in Figure \ref{fig:harp_rc_map}. \textbf{Note}: bound improvements are defined in the same way as in Figure \ref{fig:bound_improv_violin_GARP}, but the classical bilateral lower bound under homotheticity is the Paasche price index instead of the minimum price relative; the classical bilateral upper bound remains the Laspeyres price index.} 
    \label{fig:bound_improv_violin_HARP}
\end{figure}

\textbf{Appraised index error rates under homotheticity}: Figure \ref{fig:index_error_chart_HARP} repeats the appraisal of major indices based on error rates and magnitudes under homothetic preferences. Many of the findings from the non-homothetic case continue to apply when demanding homothetic tastes for the reference consumer: namely, error rates and magnitudes exhibit the same pattern between the appraised indices (all of which outperform market exchange rates); violation frequencies are lower for countries that are not used to define the reference consumer's preferences, and the magnitude of violations continues to remain small in practical terms (though is modestly higher than the magnitudes for the non-homothetic case). In contrast, there are two major differences between the rates and magnitudes shown in Figures \ref{fig:index_error_chart_HARP} and \ref{fig:index_error_chart_GARP}. First, bound violations occur with much greater frequency for all the considered indices under homothetic tastes compared to the non-homothetic case, which likely reflects the more demanding restrictions on behaviour that homotheticity requires. Second, the violations are distributed more symmetrically between the upper and lower bounds, particularly so for countries within the reference consumer set. This symmetry mirrors the similar pattern seen in the distribution of bound tightening in Figure \ref{fig:bound_improv_violin_HARP} and implies that under the chosen homothetic reference consumer, there is no longer a systematic under or over-estimation of real income or prices using existing methods. 

With respect to error rates and magnitudes under homotheticity, a final note is warranted about the interpretation of these statistics with respect to the Fisher and GEKS indices. The fact that a fifth of the bilateral index numbers in each of these cases lie outside the permitted bounds may at first glance appear inconsistent with the fact that the Fisher and GEKS indices approximate an underlying homothetic utility function. The explanation for this is that despite being homothetic, these approximated preferences are not the same as the one that defines the reference consumer's tastes.
\begin{figure}[h!]
    \centering
    \includegraphics[width=\linewidth]{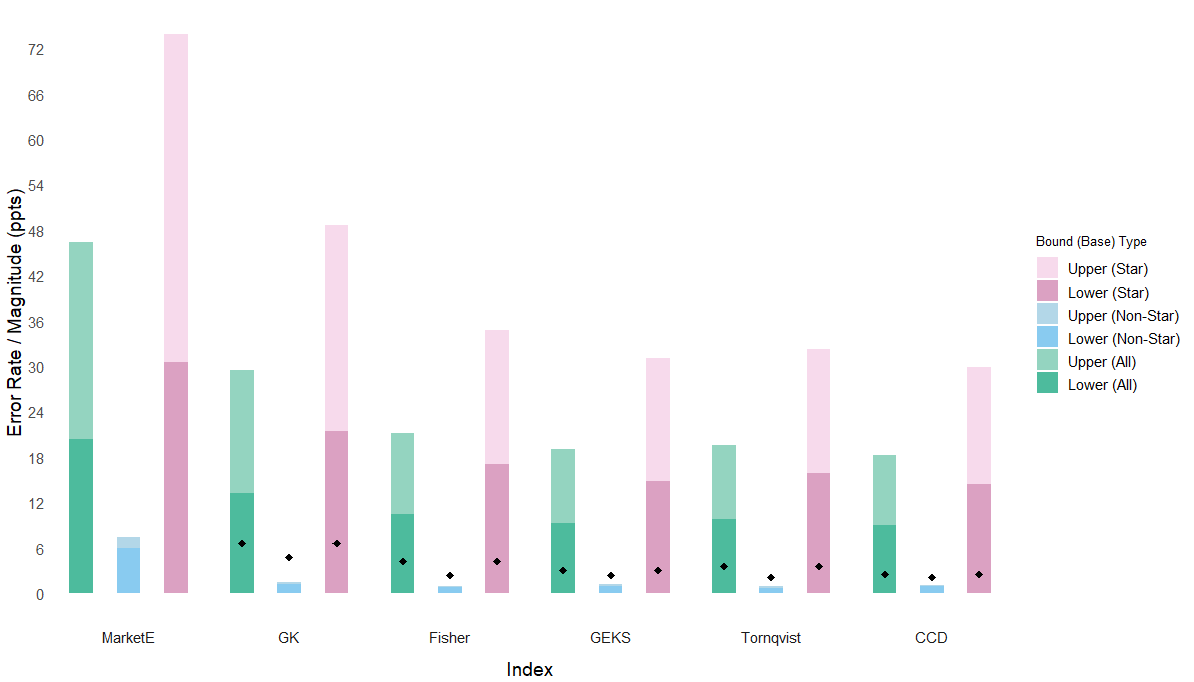}
    \caption{The average error rate $\varepsilon$ (coloured bars) and mean error magnitude $|\varepsilon|$ (black diamonds) for the five considered indices and market exchange rates based on the reference consumer tastes under homothetic preferences.} 
    \label{fig:index_error_chart_HARP}
\end{figure}

\newpage
\subsection{The relationship between the generalised star system under homotheticity, Afriat's Ideal Index, and the Fisher index when $N=2$} \label{subsec:afriat_ideal_proof}

Recall that the generalised star system in the non-homothetic case, $\sqrt{\mathcal{M}^{-}_{ij}\mathcal{M}^{+}_{ij}}$, is the geometric mean between between the lower and upper cost-of-living bounds for data which satisfy Proposition \ref{prop:ref_consumer_characterisation}.

The Ideal Afriat Index of \textcite{dowrick_true_1997} also takes the geometric mean but replaces $\mathcal{M}^{-}_{ij}$ and $\mathcal{M}^{-}_{ij}$ with their analogous bounds under homotheticity, which are sometimes known in the literature as Generalised Paasche and Laspeyres price indices (e.g. \textcite[p.~916]{manser_analysis_1988}). In the case of $N=2$, these bounds simply become the bilateral Paasche and Laspeyres price indices, which is a fact that immediately follows by inspecting the homothetic equivalent to Proposition \ref{prop:ref_consumer_characterisation} and CEWEC, also known as the Homothetic Axiom of Revealed Preference. Specifically, notice that by \eqref{eq:harp_weight_product}, for two countries $i$ and $j$ the restrictions are:
\begin{align*}
    w_{ij} & \geq 1 \\
    w_{ji} & \geq 1.
\end{align*}
Recalling that the weights are Laspeyres price indices and combining these inequalities gives the homothetic bounds and taking the geometric mean of these in turn gives the Fisher result.
\end{document}